%% file: esop21-tr.tex
\documentclass[runningheads]{llncs}

\input{macros}

\begin{document}

\title{Nested Session Types}

%
\author{Ankush Das\inst{1} \and
Henry DeYoung\inst{1} \and
Andreia Mordido\inst{2} \and 
Frank Pfenning\inst{1}}
\authorrunning{A. Das et al.}
\institute{Carnegie Mellon University, USA \and
LASIGE, Faculdade de Ci\^encias, Universidade de Lisboa, Portugal}

\maketitle


\begin{abstract}
  Session types statically describe communication protocols between
  concurrent message-passing processes. Unfortunately, parametric polymorphism
  even in its restricted prenex form
  is not fully understood in the context of session types. In this paper,
  we present the metatheory of session types extended with prenex polymorphism and,
  as a result, nested recursive datatypes.
  Remarkably, we prove that type equality is decidable by exhibiting a reduction
  to trace equivalence of deterministic first-order grammars. Recognizing the high theoretical
  complexity of the latter, we also propose a novel type equality
  algorithm and prove its soundness. We observe that the algorithm is
  surprisingly efficient and, despite its incompleteness, sufficient for
  all our examples. We have implemented our ideas by extending the Rast
  programming language with nested session types. We conclude with several examples
  illustrating the expressivity of our enhanced type system.
\end{abstract}


\input{intro}

\input{overview}

\input{types}

\input{type-equivalence}
\input{tpeq}

\input{language}

\input{context-free}

\input{implementation}

\input{examples}

\input{related-work}

\input{conclusion}

\bibliographystyle{splncs04}
\bibliography{refs}

\end{document}

%% file: macros.tex
\usepackage{stmaryrd} 
\usepackage{proof}
\usepackage{amsmath}
\usepackage{amssymb}
\usepackage{mathtools}
\usepackage{mathpartir}
\usepackage{color}
\usepackage{xstring}
\usepackage{ntabbing}
\usepackage{listings}
\usepackage{varwidth}
\usepackage{enumitem}
\usepackage{multicol}
\usepackage{lipsum}
\usepackage{booktabs}
\usepackage{subcaption}
\usepackage{fancyvrb}
\usepackage{hyperref}

\newcommand{\polymorphicsessiontypes}[0]{nested session types}

\newcommand{\Polymorphicsessiontypes}[0]{Nested session types}

\newlength{\rWidth}

\newcommand{\m}[1]{\mathsf{#1}}
\newcommand{\mb}[1]{\mathbf{#1}}

\newenvironment{sill}{\begin{tabbing}}{\end{tabbing}}

\newcommand{\Sg}{\Sigma}

\newcommand{\D}{\Delta}
\newcommand{\G}{\Gamma}

\newcommand{\proc}[2]{\m{proc}(#1, #2)}
\newcommand{\msg}[2]{\m{msg}(#1, #2)}

\newcommand{\step}{\mapsto}

\newcommand{\fresh}[1]{(#1 \text{ fresh})}

\newcommand{\ecase}[3]{\m{case} \; #1 \; (#2 \Rightarrow #3)}

\newcommand{\erecvch}[2]{#2 \leftarrow \m{recv} \; #1}

\newcommand{\esendch}[2]{\m{send} \; #1 \; #2}

\newcommand{\ewait}[1]{\m{wait} \; #1}

\newcommand{\eclose}[1]{\m{close} \; #1}

\newcommand{\fwd}[2]{#1 \leftrightarrow #2}

\newcommand{\esendl}[2]{#1.#2}

\newcommand{\ecut}[4]{#1 \leftarrow #2 \; #3 \semi #4}

\newcommand{\procdef}[3]{#3 \leftarrow #1 \; #2}

\newcommand{\lolli}{\multimap}
\newcommand{\tensor}{\otimes}
\newcommand{\with}{\mathbin{\binampersand}}

\newcommand{\one}{\mathbf{1}}

\newcommand{\semi}{\; ; \;}
\newcommand{\ichoiceop}{\oplus}
\newcommand{\echoiceop}{\with}
\newcommand{\ichoice}[1]{\mathord{\ichoiceop} \{ #1 \}}
\newcommand{\echoice}[1]{\mathord{\echoiceop} \{ #1 \}}

\newcommand{\mi}[1]{\mbox{\it #1}}

\newcommand{\tforall}[1]{\forall #1. \, }
\newcommand{\texists}[1]{\exists #1. \, }

\newcommand{\Next}{\raisebox{0.3ex}{$\scriptstyle\bigcirc$}}
\newcommand{\next}[1]{\Next #1}
\newcommand{\tdelay}[2]{
    \IfEqCase{#2}{%
        {1}{\next{#1}}%
    }[{\Next^{#2} (#1)}]%
}%

\newcommand{\indv}[1]{[\overline{#1}]}

\newcommand{\config}{\mathcal{S}}

\newcommand{\dc}{\mathcal{D}}

\newcommand{\cons}{\mathcal{C}}
\newcommand{\vars}{\mathcal{V}}

\newcommand{\unfldsg}[2]{\m{unfold}_{#1}(#2)}
\newcommand{\unfold}[1]{\unfldsg{\Sg}{#1}}
\newcommand{\rel}{\mathcal{R}}

\definecolor{verylightgray}{rgb}{.97,.97,.97}

\lstdefinelanguage{Solidity}{
	keywords=[1]{anonymous, assembly, assert, balance, break, call, callcode, case, catch, class, constant, continue, constructor, contract, debugger, default, delegatecall, delete, do, else, emit, event, experimental, export, external, false, finally, for, function, gas, if, implements, import, in, indexed, instanceof, interface, internal, is, length, library, log0, log1, log2, log3, log4, memory, modifier, new, payable, pragma, private, protected, public, pure, push, require, return, returns, revert, selfdestruct, send, solidity, storage, struct, suicide, super, switch, then, this, throw, transfer, true, try, typeof, using, value, view, while, with, addmod, ecrecover, keccak256, mulmod, ripemd160, sha256, sha3}, 
	keywordstyle=[1]\color{blue}\bfseries,
	keywords=[2]{address, bool, byte, bytes, bytes1, bytes2, bytes3, bytes4, bytes5, bytes6, bytes7, bytes8, bytes9, bytes10, bytes11, bytes12, bytes13, bytes14, bytes15, bytes16, bytes17, bytes18, bytes19, bytes20, bytes21, bytes22, bytes23, bytes24, bytes25, bytes26, bytes27, bytes28, bytes29, bytes30, bytes31, bytes32, enum, int, int8, int16, int24, int32, int40, int48, int56, int64, int72, int80, int88, int96, int104, int112, int120, int128, int136, int144, int152, int160, int168, int176, int184, int192, int200, int208, int216, int224, int232, int240, int248, int256, mapping, string, uint, uint8, uint16, uint24, uint32, uint40, uint48, uint56, uint64, uint72, uint80, uint88, uint96, uint104, uint112, uint120, uint128, uint136, uint144, uint152, uint160, uint168, uint176, uint184, uint192, uint200, uint208, uint216, uint224, uint232, uint240, uint248, uint256, var, void, ether, finney, szabo, wei, days, hours, minutes, seconds, weeks, years},	
	keywordstyle=[2]\color{teal}\bfseries,
	keywords=[3]{block, blockhash, coinbase, difficulty, gaslimit, number, timestamp, msg, data, gas, sender, sig, value, now, tx, gasprice, origin},	
	keywordstyle=[3]\color{violet}\bfseries,
	identifierstyle=\color{black},
	sensitive=false,
	comment=[l]{//},
	morecomment=[s]{/*}{*/},
	commentstyle=\color{gray}\ttfamily,
	stringstyle=\color{red}\ttfamily,
	morestring=[b]',
	morestring=[b]"
}

\newcommand{\clo}[2]{\langle #1 \semi #2\rangle}
\newcommand{\SD}{\mathcal{S}(\mathcal{D})}
\newcommand{\SDO}{\mathcal{S}(\mathcal{D}_0)}

\newcommand{\llam}[2]{\hat{\lambda} #1.\,#2}
\newcommand{\lapp}[2]{#1\,#2}

%% file: intro.tex
\section{Introduction}
\label{sec:intro}

Session types express and enforce interaction protocols in message-passing
systems~\cite{Honda93CONCUR,takeuchi1994interaction}. In this work, we focus
on \emph{binary session types} that describe bilateral protocols between two
endpoint processes performing dual actions. Binary session types obtained a
firm logical foundation since they were shown to be in a Curry-Howard correspondence
with linear logic propositions~\cite{Caires10concur,Caires16mscs,Wadler12icfp}.
This allows us to rely on properties of cut reduction to derive type safety
properties such as \emph{progress (deadlock freedom)} and \emph{preservation
(session fidelity)}, which continue to hold even when extended to
recursive types and processes~\cite{Das20FSCD}.

However, the theory of session types is still missing a crucial piece:
a general understanding of prenex (or ML-style) parametric polymorphism, encompassing
recursively defined types, polymorphic type
constructors, and nested types.  We abbreviate the sum of these
features simply as \emph{nested types}~\cite{BirdM98}.
Prior work has restricted itself to parametric polymorphism
either: in prenex form without nested types~\cite{Griffith16phd,Thiemann16icfp}; with
explicit higher-rank quantifiers~\cite{PolyESOP13,Perez14ic} (including
bounded ones~\cite{Gay08}) but without general recursion; or in
specialized form for iteration at the type
level~\cite{Thiemann20popl}.  None of these allow a free, \emph{nested} use of polymorphic type
constructors combined with prenex polymorphism.

In this paper,
we develop the metatheory of this rich language of nested session types.
Nested types are reasonably well understood in the context of
functional languages~\cite{BirdM98,Johann09hosc} and have a
number of interesting
applications~\cite{Connelly95mscs,Hinze10jfp,Okasaki96phd}.  One
difficult point is the interaction of nested types with polymorphic
recursion and type inference~\cite{Mycroft84ISP}.  By adopting
bidirectional type-checking we avoid this
particular set of problems altogether, at the cost of some additional
verbosity.  However, we have a new problem namely that session type
definitions are generally \emph{equirecursive} and \emph{not
generative}.  This means that even before we consider nesting, with
the definitions
\begin{mathpar}
  \m{list}[\alpha] = \ichoice{\mb{nil} : \one, \mb{cons} : \alpha \tensor \m{list}[\alpha]} \and
    \m{list'}[\alpha] = \ichoice{\mb{nil} : \one, \mb{cons} : \alpha \tensor \m{list'}[\alpha]}
\end{mathpar}
we have $\m{list}[A] \equiv \m{list'}[B]$ and also
$\m{list}[\m{list'}[A]] \equiv \m{list'}[\m{list}[B]]$ provided
$A \equiv B$.  The reason is that both types specify the same communication
behavior---only their name (which is irrelevant) is different.
As the second of these equalities shows, deciding the equality of
nested occurrences of type constructors is inescapable: allowing
type constructors (which are necessary in many practical examples)
means we also have to solve type equality for nested types. For example,
the types $\m{Tree}[\alpha]$ and $\m{STree}[\alpha][\kappa]$ represent
binary trees and their faithfully (and efficiently) serialized form respectively.
\begin{mathpar}
  \m{Tree}[\alpha] = \ichoice{\mb{node} : \m{Tree}[\alpha] \tensor \alpha
  \tensor \m{Tree}[\alpha] , \mb{leaf} : \one} \and
  \m{STree}[\alpha, \kappa] = \ichoice{\mb{nd} : \m{STree}[\alpha, \alpha
  \tensor \m{STree}[\alpha, \kappa]], \mb{lf} : \kappa}
\end{mathpar}
We have that $\m{Tree}[\alpha] \tensor \kappa$ is isomorphic to
$\m{STree}[\alpha, \kappa]$ and that the processes witnessing the
isomorphism can be easily implemented (see Section~\ref{sec:examples}).

%

At the core of type checking such programs lies \emph{type equality}. 
We show that we can translate type equality for nested
session types to the trace equivalence problem for
deterministic first-order grammars, which was shown to be decidable by
Jan\v{c}ar, albeit with doubly-exponential
complexity~\cite{jancar2012bisimilarity}.
Solomon~\cite{Solomon78popl} already proved a related connection between
\emph{inductive} type equality for nested types and language equality
for DPDAs.  The difference is that session type equality must be
defined coinductively, as a bisimulation, rather than via language
equivalence~\cite{Gay2005}.  This is because session types capture
communication behavior rather than the structure of closed values so a
type such as $\m{R} = \ichoice{\mb{a} : \m{R}}$ is not equal to the
empty type $\m{E} = \ichoice{}$.
The reason is that the former type can send infinitely many $\mb{a}$'s
while the latter cannot (due to the coinductive interpretation).
Interestingly, if we
imagine a lazy functional language such as Haskell with non-generative
recursive types, then $\m{R}$ and $\m{E}$ would also be different.  In
fact, nothing in our analysis of equirecursive nested types depends on
linearity, just on the coinductive interpretation of types.  Several
of our key results, namely decidability of type equality and a practical
algorithm for it, 
apply to lazy
functional languages!  Open in this different setting would still be
the question of type inference, including the treatment of polymorphic
recursion.

The decision procedure for deterministic first-order grammars does not
appear to be directly suitable for implementation, in part due to
its doubly-exponential complexity bound.  Instead we develop an
algorithm combining loop detection~\cite{Gay2005} with 
instantiation
~\cite{Das20CONCUR} and a special treatment of reflexivity
to handle all cases that would have passed in a nominal system.
The algorithm is sound, but incomplete, and reports success, a counterexample,
or an inconclusive outcome (which counts as failure).  In our
experience, the algorithm is surprisingly efficient and sufficient
for all our examples.

We have implemented nested session types and integrated them with
the Rast language that is based on session types~\cite{Das20FSCD,Das20CONCUR,Das20PPDP}.
We have evaluated our prototype on several examples such as the Dyck
language~\cite{Dyck1882}, an expression server~\cite{Thiemann16icfp}
and serializing binary trees, and standard polymorphic data
structures such as lists, stacks and queues.

Most closely related to our work is context-free session types
(CFSTs)~\cite{Thiemann16icfp}. CFSTs also enhance the expressive power
of binary session types by extending types with a notion of
sequential composition of types.
In connection with CFSTs, we identified a proper fragment of
nested session types closed under sequential composition and therefore
nested session types are strictly more expressive than CFSTs.

The main technical contributions of our work are:
\begin{itemize}
\item A uniform language of session types supporting prenex polymorphism,
  type constructors, and nested types and its type safety proof
  (Sections~\ref{sec:types}, \ref{sec:language}).
\item A proof of decidability of type equality (Section~\ref{sec:type-equiv}).
\item A practical algorithm for type equality
  and its soundness proof (Section~\ref{sec:algorithm}).
\item A proper fragment of nested session types that is
closed under sequential composition, the main feature of 
context-free session types (Section~\ref{sec:context-free}).
\item An implementation and integration with the Rast language (Section~\ref{sec:impl}).
\end{itemize}

%% file: overview.tex
\section{Overview of Nested Session Types}\label{sec:overview}

The main motivation for studying nested types is quite practical and 
generally applicable to programming languages with structural type systems.
We start by applying parametric type constructors for a standard polymorphic
queue data structure. We also demonstrate how the types can be made more
precise using nesting. A natural consequence of having nested types
is the ability to capture (communication) patterns characterized by 
context-free languages. As an illustration, we express the Dyck language of
balanced parentheses and show how nested types are connected to DPDAs also.


\paragraph{\textbf{Queues}}
A standard application of parameterized types is the
definition of polymorphic data structures such as lists, stacks, or queues.
As a simple example, consider the nested type:
\begin{mathpar}
\m{queue}[\alpha] \triangleq
  \echoice{\mb{ins} : \alpha \lolli \m{queue}[\alpha],
           \mb{del} : \ichoice{\mb{none} : \one,
                               \mb{some} : \alpha \tensor \m{queue}[\alpha]}}
\end{mathpar}
The type $\m{queue}$, parameterized by $\alpha$, represents a queue 
with values of type $\alpha$. A process providing this type offers an 
\emph{external choice} ($\with$) enabling the client to either \emph{insert} a 
value of type $\alpha$ in the queue (label $\mb{ins}$), or to 
\emph{delete} a value from the queue (label $\mb{del}$).
After receiving label $\mb{ins}$, the provider expects
to receive a value of type $\alpha$ (the $\lolli$ operator) and then proceeds
to offer $\m{queue}[\alpha]$. Upon reception of the label
$\mb{del}$, the provider queue is either empty, in which case it sends
the label $\mb{none}$ and terminates the session (as prescribed by type $\one$),
or is non-empty, in which case it sends a value of type $\alpha$ (the $\tensor$
operator) and recurses with $\m{queue}[\alpha]$.

Although parameterized type definitions are sufficient to express
the standard interface to polymorphic data structures, we propose
\emph{nested session types} which are considerably more expressive.
For instance, we can use type parameters to track the number of
elements in the queue in its type!
\begin{mathpar}
  \m{queue}[\alpha, x] \triangleq
  \echoice{\mb{ins} : \alpha \lolli \m{queue}[\alpha, \m{Some}[\alpha, x]],
           \mb{del} : x} \\
  \m{Some}[\alpha, x] \triangleq \ichoice{\mb{some} : \alpha \tensor \m{queue}[\alpha, x] } \and
  \m{None} \triangleq \ichoice{\mb{none} : \one }
\end{mathpar}
The second type parameter $x$ tracks the number of
elements. This parameter can be understood as a \emph{symbol
stack}. On inserting an element, we recurse
to $\m{queue}[\alpha, \m{Some}[\alpha, x]]$ denoting the \emph{push} of $\m{Some}$ symbol on
stack $x$. We initiate the empty queue with the type $\m{queue}[\alpha, \m{None}]$
where the second parameter denotes \emph{an empty symbol stack}.
Thus, a queue with $n$ elements would have the type $\m{queue}[\alpha, \m{Some}^n[\alpha, \m{None}]]$.
On receipt of the $\m{del}$ label, the type transitions to $x$ which can
either be $\m{None}$ (if the queue is empty) or $\m{Some}[\alpha, x]$ (if
the queue is non-empty). In the latter case, the type sends label $\mb{some}$
followed by an element, and transitions to $\m{queue}[\alpha, x]$ denoting a
\emph{pop} from the symbol stack. In the former case, the type sends the label $\m{none}$
and terminates. Both these behaviors are reflected in the definitions of types
$\m{Some}$ and $\m{None}$.

\paragraph{\textbf{Context-Free Languages}}
Recursive session types 
capture the class of regular languages~\cite{Thiemann16icfp}. 
However, in practice, many useful languages are beyond regular. 
As an illustration, suppose we would like to express a balanced parentheses
language, also known as the Dyck language~\cite{Dyck1882} with the
end-marker $\$$.
We use $\mb{L}$ to denote an opening
symbol, and $\mb{R}$ to denote a closing symbol
(in a session-typed mindset, $\mb{L}$ can represent client request
and $\mb{R}$ is server response). We need to enforce
that each $\mb{L}$ has a corresponding closing $\mb{R}$ and they are
properly nested. 
To express this, we need to track the
number of $\mb{L}$'s in the output with the session type. However,
this notion of \emph{memory} is beyond the expressive power of regular languages,
so mere recursive session types will not suffice.

We utilize the expressive power of nested types to express this behavior.
\begin{center}
  \begin{minipage}{0cm}
  \begin{tabbing}
  $T[x] \triangleq \ichoice{\mb{L} : T[T[x]], \mb{R} : x} \qquad
  D \triangleq \ichoice{\mb{L} : T[D], \mb{\$} : \one}$
  \end{tabbing}
  \end{minipage}
\end{center}
The nested type $T[x]$ takes $x$ as a type parameter and either outputs
$\mb{L}$ and continues with $T[T[x]]$, or outputs $\mb{R}$ and continues
with $x$. The type $D$ either outputs $\mb{L}$ and continues
with $T[D]$, or outputs $\mb{\$}$ and terminates. The type $D$ expresses
a Dyck word with end-marker $\$$~\cite{korenjak1966simple}.

The key idea here is that the number of $T$'s in the type of a word
tracks the number of unmatched $\mb{L}$'s in it. Whenever the type $T[x]$
outputs $\mb{L}$, it recurses with $T[T[x]]$ incrementing the number of
$T$'s in the type by $1$. Dually, whenever the type outputs $\mb{R}$, it
recurses with $x$ decrementing the number of $T$'s in the type by $1$.
The type $D$ denotes a balanced word with no unmatched $\mb{L}$'s. 
Moreover, since we can only output $\$$
(or $\mb{L}$) at the type $D$ and \emph{not} $\mb{R}$, we obtain the invariant that
any word of type $D$ must be balanced.
If we imagine the parameter $x$ as the symbol stack, outputting an $\mb{L}$
pushes $T$ on the stack, while outputting $\mb{R}$ pops $T$ from the stack.
The definition of $D$ ensures that once an $\mb{L}$ is outputted, the symbol
stack is initialized with $T[D]$ indicating one unmatched $\mb{L}$.

Nested session types do not restrict communication so that the words
represented \emph{have to be balanced}. To this end,
the type $D'$ can model the \emph{cropped Dyck language}, where
\emph{unbalanced} words can be captured.
\begin{mathpar}
  T'[x] \triangleq \ichoice{\mb{L} : T'[T'[x]], \mb{R} : x, \mb{\$} : \one} \and
  D' \triangleq \ichoice{\mb{L} : T'[D'], \mb{\$} : \one}
\end{mathpar}
The only difference between types $T[x]$ and $T'[x]$ is that $T'[x]$
allows us to terminate at any point using the $\$$ label which
immediately transitions to type $\one$.
\Polymorphicsessiontypes\ can not only capture
the class of deterministic context-free languages recognized by
DPDAs that \emph{accept by empty stack} (balanced words), but also
the class of deterministic context-free languages recognized  
by DPDAs that \emph{accept by final state} (cropped words).

\paragraph{\textbf{Multiple Kinds of Parentheses}}
We can use nested types to express more general words
with different kinds of parentheses. Let $\mb{L}$ and $\mb{L'}$
denote two kinds of opening symbols, while $\mb{R}$ and $\mb{R'}$ denote
their corresponding closing symbols respectively. We define the session types
\[
  \begin{array}{rcl}
  S[x] & \triangleq & \ichoice{\mb{L} : S[S[x]], \mb{L'} : S'[S[x]], \mb{R} : x} \\[0.2em]
  S'[x] & \triangleq & \ichoice{\mb{L} : S[S'[x]], \mb{L'} : S'[S'[x]], \mb{R'} : x} \\[0.2em]
  E & \triangleq & \ichoice{\mb{L} : S[E], \mb{L'} : S'[E], \mb{\$} : \one}
  \end{array}
\]
We \emph{push} symbols $S$ and
$S'$ to the stack on outputting $\mb{L}$ and $\mb{L'}$ respectively.
Dually, we \emph{pop} $S$ and $S'$ from the stack on outputting
$\mb{R}$ and $\mb{R'}$ respectively. Then, the type $E$ defines
an \emph{empty stack}, thereby representing a balanced Dyck word.
This technique can be generalized to any number of kinds of
brackets.



\paragraph{\textbf{Multiple States as Multiple Parameters}}
Using defined type names with \emph{multiple} type parameters, we
enable types to capture the language of DPDAs with several states.
Consider the language
$L_3 = \{\mb{L} ^n \mb{a} \,\mb{R}^n \mb{a} \cup \mb{L}^n \mb{b} \,\mb{R}^n \mb{b} \mid n > 0\}$, 
proposed by Korenjak and Hopcroft~\cite{korenjak1966simple}. A word in this language starts with a
sequence of opening symbols $\mb{L}$, followed by an \emph{intermediate symbol},
either $\mb{a}$ or $\mb{b}$. Then, the word contains as many closing symbols
$\mb{R}$ as there were $\mb{L}$s and terminates with the symbol $\mb{a}$ or $\mb{b}$
\emph{matching} the intermediate symbol.
\begin{mathpar}
  U \triangleq \ichoice{\mb{L} : O[C[A],C[B]]} \and
  O[x,y] \triangleq \ichoice{\mb{L} : O[C[x],C[y]], \mb{a} : x, \mb{b} : y} \\
  C[x] \triangleq \ichoice{\mb{R} : x} \and
  A \triangleq \ichoice{\mb{a} : \one} \and
  B \triangleq \ichoice{\mb{b} : \one} \and
\end{mathpar}
The $L_3$ language is characterized by session type $U$. Since the type $U$
is unaware of which intermediate symbol among $\mb{a}$ or $\mb{b}$ would eventually
be chosen, it cleverly maintains \emph{two symbol stacks} in the two type
parameters $x$ and $y$ of $O$. We initiate type $U$ with outputting $\mb{L}$
and transitioning to $O[C[A], C[B]]$ where the symbol $C$ tracks that we
have outputted \emph{one} $\mb{L}$. The types $A$ and $B$ represent the
intermediate symbols that might be used in the future. The type $O[x,y]$
can either output an $\mb{L}$ and transition to $O[C[x],C[y]]$ \emph{pushing}
the symbol $C$ onto \emph{both} stacks; or it can output $\mb{a}$ (or $\mb{b}$) and transition
to the first (resp. second) type parameter $x$ (resp. $y$). Intuitively,
the type parameter $x$ would have the form $C^n[A]$ for $n > 0$ (resp. $y$
would be $C^n[B]$). Then, the type $C[x]$ would output an $\mb{R}$ and \emph{pop}
the symbol $C$ from the stack by transitioning to $x$. Once all the closing
symbols have been outputted (note that you cannot terminate pre-emptively),
we transition to type $A$ or $B$ depending on the intermediate symbol chosen.
Type $A$ outputs $\mb{a}$ and terminates, and similarly, type $B$ outputs $\mb{b}$
and terminates. Thus, we simulate the $L_3$ language (not possible
with context-free session types~\cite{Thiemann16icfp}) using two type parameters.

More broadly, nested types can neatly capture \emph{complex server-client interactions}.
For instance, client requests can be captured using labels $\mb{L, L'}$ while
server responses can be captured using labels $\mb{R, R'}$ expressing \emph{multiple
kinds} of requests. Balanced words will then represent that all requests
have been handled. The types can also guarantee that responses do not 
exceed requests.

\paragraph{\textbf{Concatenation of Dyck Words}}
We conclude this section by proving some standard properties on balanced
parentheses: \emph{closure under concatenation} and \emph{closure under
wrapping}. If $w_1\$$ and $w_2\$$
are two balanced words, then so is $w_1 w_2 \$$. Similarly, if
$w\$$ is a balanced word, then so is $\mb{L}w\mb{R}\$$. These two properties
can be proved by implementing $\mi{append}$ and $\mi{wrap}$ processes
capturing the former and latter properties.
\begin{mathpar}
  \mi{append} : (w_1 : D), (w_2 : D) \vdash (w : D) \and
  \mi{wrap} : (w : D) \vdash (w' : D)
\end{mathpar}
The above declarations describe the type for the two processes.
The $\mi{append}$ process uses two channels $w_1$ and $w_2$
of type $D$ and provides $w : D$, whereas $\mi{wrap}$ uses $w : D$
and provides $w' : D$.

\begin{Verbatim}
  decl fmap'[a][b] : (f : a -o b) |- (g : T[a] -o T[b])
  proc g <- fmap'[a][b] f =
    w <- recv g ; % (f : a -o b) (w : T[a]) |- (g : T[b])
    case w (
      L =>        % (f : a -o b) (w : T[T[a]]) |- (g : T[b])
        g.L ;     % (f : a -o b) (w : T[T[a]]) |- (g : T[T[b]])
        h0 <- fmap'[a][b] f ;
        h1 <- fmap'[T[a]][T[b]] h0 ;
        send h1 w ; g <-> h1
    | R =>    % (f : a -o b) (w : a) |- (g : T[b])
        g.R ; % (f : a -o b) (w : a) |- (g : b)
        send f w ; g <-> f
    )

  decl fmap[a][b] : (f : a -o b) (w : T[a]) |- (w' : T[b])
  proc w' <- fmap[a][b] f w =
    f' <- fmap'[a][b] f ; % (f' : T[a] -o T[b]) (w : T[a]) |- (w' : T[b])
    send f' w ; w' <-> f'

  decl append' : (w2 : D) |- (f : D -o D)
  proc f <- append' w2 =
    w1 <- recv f ; % (w1 : D) (w2 : D) |- (f : D)
    case w1 (
      L => % (w1 : T[D]) (w2 : D) |- (f : D)
        f.L ; % (w1 : T[D]) (w2 : D) |- (f : T[D])
        g <- append' w2 ; % (w1 : T[D]) (g : D -o D) |- (f : T[D])
        f <- fmap[D][D] g w1
    | $ => % (w1 : 1) (w2 : D) |- (f : D)
        wait w1 ; f <-> w2
    )

  proc w <- append w1 w2 =
    f <- append' w2 ; % (w1 : D) (f : D -o D) |- (w : D)
    send f w1 ; w <-> f
\end{Verbatim}

%% file: types.tex
\section{Description of Types}
\label{sec:types}

The underlying base system of session types is derived from a Curry-Howard
interpretation~\cite{Caires10concur,Caires16mscs} of intuitionistic linear logic
\cite{Girard87tapsoft}.
Below we describe the session types, their operational interpretation
and the continuation type.
\[\arraycolsep=4.5pt
  \begin{array}{lclll}
    A,B,C & ::= & \ichoice{\ell : A_\ell}_{\ell \in L} & \mbox{send label $k \in L$} & \mbox{continue at type $A_k$} \\
          & \mid & \echoice{\ell : A_\ell}_{\ell \in L} & \mbox{receive label $k \in L$} & \mbox{continue at type $A_k$} \\
          & \mid & A \tensor B & \mbox{send channel $a : A$} & \mbox{continue at type $B$} \\
          & \mid & A \lolli B & \mbox{receive channel $a : A$} & \mbox{continue at type $B$} \\
          & \mid & \one & \mbox{send $\m{close}$ message} & \mbox{no continuation} \\
          & \mid & \alpha & \mbox{type variable} \\
          & \mid & V\indv{B} & \mbox{defined type name} \\
  \end{array}
\]

The basic type operators have the usual interpretation:
the \emph{internal choice} operator $\ichoice{\ell \colon A_{\ell}}_{\ell\in L}$ 
selects a branch with label $k \in L$ with corresponding
continuation type $A_k$; the \emph{external choice} operator 
$\echoice{\ell \colon A_{\ell}}_{\ell\in L}$ offers a choice
with labels $\ell \in L$ with corresponding continuation types
$A_{\ell}$; the \emph{tensor} operator $ A \tensor B$ 
represents the channel passing type that consists of sending a
channel of type $A$ and proceeding with type $B$;
dually, the \emph{lolli} operator $A \lolli B$ consists of 
receiving a channel of type $A$ and continuing with 
type $B$; the \emph{terminated session} $\one$ is
the operator that closes the session.

We also support \emph{type constructors} to define new \emph{type names}. 
A type name $V$ is defined according to a \emph{type definition} 
$V\indv{\alpha} = A$ that is parameterized by a sequence of
\emph{distinct type variables} $\overline\alpha$ that the type $A$ can refer to.
We can use type names in a type expression using $V \indv{B}$.
Type expressions can also refer to parameter $\alpha$
available in scope.
The \emph{free variables} in type $A$ refer to the set of type variables
that occur freely in $A$.
Types without any free variables are called \emph{closed types}.
We call any type not of the form $V \indv{B}$ to be \emph{structural}.

All type definitions are stored in a finite global \emph{signature} $\Sg$
defined as
\[
  \begin{array}{llcl}
    \mbox{Signature} & \Sigma & ::= & \cdot \mid \Sigma, V\indv{\alpha} = A
  \end{array}
\]
In a \emph{valid signature}, all definitions $V\indv{\alpha} = A$ are contractive,
meaning that $A$ is \emph{structural}, i.e. not itself a type name.
This allows us to take an \emph{equirecursive} view of type definitions, which means
that unfolding a type definition does not require communication.
More concretely, the type $V \indv{B}$ is considered equivalent to its unfolding
$A[\overline{B}/\overline{\alpha}]$.
We can easily adapt our definitions to an \emph{isorecursive}
view~\cite{Lindley16icfp,Derakhshan19corr} with explicit unfold messages.
All type names $V$ occurring in a valid signature
must be defined, and all type variables defined in a valid
definition must be distinct.
Furthermore, for a valid definition $V \indv{\alpha} = A$, the free variables occurring
in $A$ must be contained in $\overline{\alpha}$.
This top-level scoping of all type variables is what
we call the \emph{prenex form of polymorphism}.

%% file: type-equivalence.tex
\section{Type Equality}
\label{sec:type-equiv}

Central to any practical type checking algorithm is type equality.
In our system, it is necessary for the rule of identity (forwarding)
and process spawn, as well as the channel-passing constructs
for types $A \tensor B$ and $A \lolli B$.
However, with nested polymorphic recursion, checking equality becomes
challenging.
We first develop the underlying theory of equality providing its
definition, and then establish its reduction to checking trace
equivalence of deterministic first-order grammars.


\subsection{Type Equality Definition}
\label{subsec:type-equality}
Intuitively, two types are equal if they permit exactly the \emph{same}
communication behavior.
Formally, type equality is captured using a coinductive definition
following seminal work by Gay and Hole~\cite{Gay2005}.
\begin{definition}\label{def:unfold}
  We first define $\unfold{A}$ as
  \begin{mathpar}
    \infer[\m{def}]
    {\unfold{V \indv{B}} = A[\overline{B}/\overline{\alpha}]}
    {V \indv{\alpha} = A \in \Sg}
    \and
    \infer[\m{str}]
    {\unfold{A} = A}
    {A \not= V \indv{B}}
  \end{mathpar}
\end{definition}

Unfolding a structural type simply returns $A$.
Since type definitions are \emph{contractive}~\cite{Gay2005},
the result of unfolding is never a type name application and it always
terminates in one step.
\begin{definition}\label{def:rel}
	Let $\mi{Type}$ be the set of closed type expressions (no free variables).
  A relation $\rel \subseteq
  \mi{Type} \times \mi{Type}$ is a type bisimulation if $(A, B) \in
  \rel$ implies:
  \begin{itemize}[leftmargin=*]
    \item If $\unfold{A} = \ichoice{\ell : A_\ell}_{\ell \in L}$, then $\unfold{B} =
    \ichoice{\ell : B_\ell}_{\ell \in L}$ and also $(A_\ell, B_\ell) \in \rel$ for
    all $\ell \in L$.

    \item If $\unfold{A} = \echoice{\ell : A_\ell}_{\ell \in L}$, then $\unfold{B} =
    \echoice{\ell : B_\ell}_{\ell \in L}$ and also $(A_\ell, B_\ell) \in \rel$ for
    all $\ell \in L$.

    \item If $\unfold{A} = A_1 \tensor A_2$, then $\unfold{B} =
    B_1 \tensor B_2$ and $(A_1, B_1) \in \rel$ and
    $(A_2, B_2) \in \rel$.

    \item If $\unfold{A} = A_1 \lolli A_2$, then $\unfold{B} =
    B_1 \lolli B_2$ and $(A_1, B_1) \in \rel$ and
		$(A_2, B_2) \in \rel$.
		
    \item If $\unfold{A} = \one$, then $\unfold{B} = \one$.
  \end{itemize}
\end{definition}

\begin{definition}\label{def:tpeq}
  Two closed types $A$ and $B$ are equal ($A \equiv B$) iff there exists a type
  bisimulation $\rel$ such that $(A, B) \in \rel$.
\end{definition}

When the signature $\Sigma$ is not clear from context we add a subscript,
$A \equiv_{\Sigma} B$.
This definition only applies to types with no free type variables.
Since we allow parameters in type definitions,
we need to define equality in the presence of free type variables.
To this end, we define the notation $\forall \vars. \, A \equiv B$
where $\vars$ is a collection of type variables and $A$ and $B$
are valid types w.r.t. $\vars$ (i.e., free variables in $A$ and
$B$ are contained in $\vars$).
\begin{definition}\label{def:tpeq_vars}
  We define $\forall \vars. \, A \equiv B$ iff for all
  closed type substitutions $\sigma : \vars$,
  we have $A[\sigma] \equiv B[\sigma]$.
\end{definition}

\subsection{Decidability of Type Equality}

Solomon~\cite{Solomon78popl} proved that types defined using parametric
type definitions with an \emph{inductive interpretation} can be
translated to DPDAs, thus reducing type equality
to language equality on DPDAs. However, our type definitions  
have a \emph{coinductive interpretation}.
As an example, consider the types 
$A= \ichoice{\mb{a}: A}$ and
$B= \ichoice{\mb{b}: B}$. With an \emph{inductive} interpretation, 
types $A$ and $B$ are empty (because they do not have terminating symbols) and, thus, are equal.
However, with a \emph{coinductive} interpretation, type $A$ will
send an infinite number of $\mb{a}$'s, and $B$ will send an infinite
number of $\mb{b}$'s, and are thus not equal.
Our reduction needs to account for this coinductive behavior.

We show that type equality of \polymorphicsessiontypes\ is decidable
via a reduction to the trace equivalence problem of deterministic
first-order grammars~\cite{Jancar10}.
A \emph{first-order grammar} is a structure $(\mathcal{N},
\mathcal{A},\mathcal{S})$ where 
$\mathcal{N}$ is a set of non-terminals,
$\mathcal{A}$ is a finite set of \emph{actions}, 
and $\mathcal{S}$ is a finite set of \emph{production rules}.
The arity of non-terminal $X\in \mathcal{N}$ is written as $\mathsf{arity}(X)\in\mathbb{N}$.
Production rules rely on a countable set of \emph{variables} $\vars$,
and on the set $\mathcal{T}_\mathcal{N}$ of \emph{regular terms} over
$\mathcal{N}\cup \vars$.
A term is \emph{regular} if the set of subterms is finite
(see~\cite{Jancar10}).

Each production rule has the form $X \overline{\alpha} \xrightarrow{a}
E$ where $X \in \mathcal{N}$ is a non-terminal, $a \in \mathcal{A}$ is an action,
and $\overline{\alpha} \in \vars^*$ are variables that the term 
$E\in \mathcal{T}_\mathcal{N}$ can refer to.
A grammar is \emph{deterministic} if for each pair of $X\in \mathcal{N}$ and 
$a \in\mathcal{A}$,
there is at most one rule of the form 
$X \overline\alpha\xrightarrow{a} E$
in $\mathcal{S}$.
 The substitution of terms $\overline{B}$
 for variables $\overline\alpha$ in a rule 
 $X \overline\alpha\xrightarrow{a} E$,
 denoted by
 $X \overline B\xrightarrow{a} E [\overline B / \overline \alpha]$,
 is the rule
 $(X \overline\alpha\xrightarrow{a} E)[\overline B / \overline \alpha]$.
Given a set of rules $\mathcal{S}$, the trace of a term $T$ is defined as
$\m{trace}_{\mathcal{S}}(T) = \{\overline{a} \in \mathcal{A}^* \mid 
(T \xrightarrow{\overline{a}} T') \in \mathcal{S}, 
\text{ for some }T'\in \mathcal{T}_\mathcal{N} \}$.
Two terms are \emph{trace equivalent}, written as 
$T \sim_{\mathcal{S}} T'$,
if $\m{trace}_{\mathcal{S}}(T) = \m{trace}_{\mathcal{S}}(T')$. 

\begin{figure}[t]
\begin{mathpar}
  \infer[\oplus]
	{\vars \vdash \ichoice{\ell : A_\ell}_{\ell \in L} \rightarrow
	(\ichoice{\ell : B_\ell}_{\ell \in L}, \cup_{\ell \in L} \Sg_\ell)}
  {\vars \vdash A_\ell \Rightarrow (B_\ell, \Sg_\ell) \quad (\forall \ell \in L)}
  \and
  \infer[\with]
	{\vars \vdash \echoice{\ell : A_\ell}_{\ell \in L} \rightarrow
	(\echoice{\ell : B_\ell}_{\ell \in L}, \cup_{\ell \in L} \Sg_\ell)}
  {\vars \vdash A_\ell \Rightarrow (B_\ell, \Sg_\ell) \quad (\forall \ell \in L)}
  \and
  \infer[\tensor]
  {\vars \vdash A_1 \tensor A_2 \rightarrow (B_1 \tensor B_2, \Sg_1 \cup \Sg_2)}
  {\vars \vdash A_1 \Rightarrow (B_1, \Sg_1) \and
  \vars \vdash A_2 \Rightarrow (B_2, \Sg_2)}
  \and
  \infer[\lolli]
  {\vars \vdash A_1 \lolli A_2 \rightarrow (B_1 \lolli B_2, \Sg_1 \cup \Sg_2)}
  {\vars \vdash A_1 \Rightarrow (B_1, \Sg_1) \and
  \vars \vdash A_2 \Rightarrow (B_2, \Sg_2)}
  \\\\
  \infer[\one]
  {\vars \vdash \one \rightarrow (\one, \cdot)}
  {}
  \and
  \infer[\m{var}]
  {\vars \vdash \alpha \rightarrow (\alpha, \cdot)}
	{}
	\\\\
	\infer[\m{rename-str}]
	{\vars \vdash A \Rightarrow (B, \Sg \,@\,V[\vars] = B)}
	{A \; \m{structural} \and \vars \vdash A \rightarrow (B, \Sg)
	\and (V \; \m{fresh})}
	\and
	\infer[\m{rename-nostr}]
	{\vars \vdash A \Rightarrow (A, \cdot)}
	{A = \one,\alpha}
	\\\\
	\infer[\m{emp}]
	{(\cdot) \longrightarrow (\cdot)}
	{}
	\and
	\infer[\m{step}]
	{\Sg, V\indv{\alpha} = A \longrightarrow \Sg', \Sg_A, V\indv{\alpha} = B}
	{\Sg \longrightarrow \Sg' \quad
	\overline{\alpha} \vdash A \rightarrow (B, \Sg_A)}
\end{mathpar}
\caption{Algorithmic Rules for Internal Renaming}
\label{fig:renaming_rules}
\end{figure}

The crux of the reduction lies in the observation that session types
can be translated to terms and type definitions
can be translated to production rules of a first-order grammar.
We start the translation of nested session types to grammars 
by first making an initial pass over
the signature and introducing fresh \emph{internal names}
such that the new type definitions alternate between
structural (except $\one$ and $\alpha$) and non-structural types. These internal names are
parameterized over their free type variables, and their
definitions are added to the signature.
This \emph{internal renaming} simplifies the next step where
we translate this extended signature to grammar production
rules. 

The internal renaming is defined using the judgment $\Sg \longrightarrow \Sg'$
as defined in Figure~\ref{fig:renaming_rules}.
Each definition is taken from the signature $\Sg$, and then
the definition is internally renamed and added to the original
signature.

\begin{example}\label{ex:queue_grammars}
	As a running example, consider the queue type from 
	Section~\ref{sec:overview}:
  \begin{mathpar}
  	Q[\alpha] = \echoice{\mb{ins} : \alpha \lolli Q[\alpha],
  	\mb{del} : \ichoice{\mb{none} : \one, \mb{some} : \alpha \tensor Q[\alpha]}}
	\end{mathpar} 
\end{example}

After performing internal renaming for this type,
we obtain the following signature: 
\[
	\begin{array}{ll}
		Q[\alpha] = {\with}\{\mb{ins} : X_0[\alpha], \mb{del} : X_1[\alpha]\} \hspace{2em} &
		X_1[\alpha] = {\oplus}\{\mb{none} : \one, \mb{some} : X_2[\alpha]\} \\[0.3em]
		X_0[\alpha] = \alpha \lolli Q[\alpha] &
		X_2[\alpha] = \alpha \tensor Q[\alpha]
	\end{array}
\]
We introduce the fresh internal names $X_0$, $X_1$ and $X_2$ (parameterized
with free variable $\alpha$) to represent the continuation type in each case.
Note the alternation between structural and non-structural types (of the
form $V \indv{B}$).

Next, we translate this extended signature to the grammar
$\mathcal{G} = (\mathcal{N},\mathcal{A},\mathcal{S})$ aimed
at reproducing the behavior prescribed by the types as grammar actions.
\[
 \begin{array}{rcll}
	\mathcal{N} & = & \{Q, X_0, X_1, X_2, \bot \}\smallskip\\
	\mathcal{A} & = & \{ \& \mb{ins},\& \mb{del}, \lolli_1, \lolli_2
	 \ichoiceop \mb{none}, \ichoiceop \mb{some},\tensor_1, \tensor_2,
	 \}\smallskip\\
	\mathcal{S} & = & \{ 
	Q\alpha \xrightarrow{\& \mb{ins}} X_0\alpha, \enspace 
	Q\alpha \xrightarrow{\& \mb{del}} X_1\alpha, \enspace 
	X_0\alpha \xrightarrow{\lolli_1 } \alpha, \enspace
	X_0\alpha \xrightarrow{\lolli_2 } Q\alpha, \enspace\\
	&&X_1\alpha \xrightarrow{\ichoiceop \mb{none} } \bot, \enspace 
	X_1\alpha \xrightarrow{\ichoiceop \mb{some} } X_2\alpha, \enspace 
	X_2\alpha \xrightarrow{\tensor_1 } \alpha, \enspace 
	X_2\alpha \xrightarrow{\tensor_2 } Q\alpha
	\}
  \end{array}
\]
Essentially, each defined type name is translated to a fresh non-terminal.
Each type definition then corresponds a sequence of rules: one
for each possible continuation type with the appropriate label
that leads to that continuation.
For instance, the type $Q[\alpha]$ has two possible continuations: transition
to $X_0[\alpha]$ with action $\& \mb{ins}$ or to $X_1[\alpha]$ with
action $\& \mb{del}$.
The rules for all other type names is analogous.
When the continuation is $\one$, we transition to the nullary 
non-terminal $\bot$ disabling any further action.
When the continuation is $\alpha$, we transition to $\alpha$.
Since each type name is defined once, the produced grammar is
deterministic.

Formally, the translation from an (extended) signature to
a grammar is handled by two simultaneous tasks: translating type definitions
into production rules (function $\tau$ below), and converting type names,
variables and the terminated session into grammar terms (function 
$\llparenthesis\cdot\rrparenthesis$). 
The function
$\llparenthesis \cdot \rrparenthesis : \mi{OType} \to \mathcal{T}_{\mathcal{N}}$
from open session types to grammar terms is defined by:
\[
 \begin{array}{rcll}
	\llparenthesis\one\rrparenthesis & = & \bot & \mbox{type $\one$ translates to $\bot$}\\
	\llparenthesis\alpha\rrparenthesis & = & \alpha & \mbox{type variables translate to themselves}\\
	\llparenthesis V[B_1, \ldots, B_n]\rrparenthesis & = & 
	V \llparenthesis B_1\rrparenthesis \cdots \llparenthesis B_n\rrparenthesis \quad &
	\mbox{type names translate to first-order terms}
   \end{array}
\]
Due to this mapping, throughout this section we will use type names 
indistinctly as type names or as non-terminal first-order symbols.

The function $\tau$ converts a type definition $V \indv{\alpha} = A$ 
into a set of production rules and is defined according to the structure 
of $A$ as follows:
\begin{equation*}
	\begin{array}{rcl}
			\tau(V[\overline\alpha] = \ichoice{\ell \colon A_\ell}_{\ell\in L})& = 
			& \{ \llparenthesis V[\overline\alpha]\rrparenthesis
			\xrightarrow{\ichoiceop \ell} 
			\llparenthesis A_\ell\rrparenthesis
			\mid \ell\in L
			\}
			\\
			\tau(V[\overline\alpha] = \echoice{\ell \colon A_\ell}_{\ell\in L})& = 
			& \{ \llparenthesis V[\overline\alpha]\rrparenthesis
			\xrightarrow{\echoiceop \ell} 
			\llparenthesis A_\ell\rrparenthesis
			\mid \ell\in L
			\}
			\\
			\tau(V[\overline\alpha] = A_1\tensor A_2)& = 
			& \{ \llparenthesis V[\overline\alpha]\rrparenthesis
			\xrightarrow{\tensor_i} 
			\llparenthesis A_i\rrparenthesis
			\mid i = 1,2
			\}
			\\
			\tau(V[\overline\alpha] = A_1\lolli A_2)& = 
			& \{ \llparenthesis V[\overline\alpha]\rrparenthesis
			\xrightarrow{\lolli_i} 
			\llparenthesis A_i\rrparenthesis
			\mid i = 1,2
			\}
	 \end{array}
\end{equation*}
The function $\tau$ identifies the actions and continuation types
corresponding to $A$ and translates them to grammar rules.
Internal and external choices lead to actions
$\ichoiceop \ell$ and $\echoiceop \ell$, for each $\ell \in L$,
with $A_\ell$ as the continuation type. The type $A_1\tensor A_2$ enables
two possible actions, $\tensor_1$ and $\tensor_2$, with continuation
$A_1$ and $A_2$ respectively. Similarly $A_1\lolli A_2$ produces
the actions $\lolli_1$ and $\lolli_2$ with $A_1$ and $A_2$
as respective continuations.
Contractiveness ensures that there are no definitions of the
form $V\indv{\alpha} = V'\indv{B}$.
Our internal renaming ensures that we do not encounter cases of
the form $V\indv{\alpha} = \one$ or $V \indv{\alpha} = \alpha$
because we do not generate internal names for them.
For the same reason, the $\llparenthesis\cdot\rrparenthesis$
function is only defined on the complement types $\one$, $\alpha$
and $V\indv{B}$.

The $\tau$ function is extended to translate a signature
by being applied point-wise. Formally, $\tau(\Sigma) =
\bigcup_{(V[\overline\alpha] = A)\in \Sigma}
\tau(V[\overline\alpha] = A)$.
Connecting all pieces, we define the $\m{fog}$
function that translates a signature to a grammar as:
\[
 \begin{array}{l}
 \m{fog}(\Sigma) = (\mathcal{N}, \mathcal{A}, \mathcal{S}),\text{ where: } \qquad
 \mathcal{S}  = \tau(\Sigma) \\[0.3em]
	\quad
	\mathcal{N}  = 
	\{ X\mid (X\overline{\alpha} \xrightarrow{a} 
	E) \in \tau(\Sigma)	\} \qquad
	\mathcal{A}  = 
	\{ a\mid (X\overline{\alpha} \xrightarrow{a} 
	E) \in \tau(\Sigma)	\}
  \end{array}
\]
The grammar is constructed by first computing $\tau(\Sg)$
to obtain all the production rules. $\mathcal{N}$ and
$\mathcal{A}$ are constructed by collecting the set of
non-terminals and actions from these rules.
The finite representation of session types and uniqueness
of definitions ensure that $ \m{fog}(\Sigma)$ is a 
deterministic first-order grammar.

Checking equality of types $A$ and $B$ given signature
$\Sg$ finally reduces to (i) internal renaming of $\Sg$
to produce $\Sg'$, and (ii) checking trace-equivalence of
terms $\llparenthesis A\rrparenthesis$ and 
$\llparenthesis B\rrparenthesis$ given grammar $\m{fog}(\Sg')$.
If $A$ and $B$ are themselves structural, we generate
internal names for them also during the internal renaming process.
Since we assume an \emph{equirecursive} and \emph{non-generative}
view of types, it is easy to show that internal renaming does not
alter the communication behavior of types and preserves type
equality. Formally, $A \equiv_{\Sg} B$ iff $A \equiv_{\Sg'} B$.

\begin{theorem}
\label{thm:equiv-grammars}
	${A}\equiv_{\Sg}{B}$\ if and only if 
	$\llparenthesis A\rrparenthesis \sim_{\mathcal{S}}
	\llparenthesis B\rrparenthesis$,
	where
	$(\mathcal{N}, \mathcal{A}, \mathcal{S}) = \m{fog}(\Sigma')$
	and $\Sigma'$ is the extended signature for $\Sg$.
\end{theorem}

\begin{proof}
	For the direct implication, assume that 
	$\llparenthesis A\rrparenthesis \not\sim_{\mathcal{S}}
	\llparenthesis B\rrparenthesis$.
	Pick a sequence of actions in the difference of the traces and let
	$w_0$ be its greatest prefix occurring in both traces.
	Either $w_0$ is a maximal trace for one of the terms,
	or we have 
	$\llparenthesis A\rrparenthesis \xrightarrow{w_0}\llparenthesis A'\rrparenthesis$
	and 
	$\llparenthesis B\rrparenthesis \xrightarrow{w_0}\llparenthesis B'\rrparenthesis$,
	with 
	$\llparenthesis A'\rrparenthesis \xrightarrow{a_1}\llparenthesis A''\rrparenthesis$
	and
	$\llparenthesis B'\rrparenthesis \xrightarrow{a_2}\llparenthesis B''\rrparenthesis$,
	where $a_1\neq a_2$. In both cases, we have $A' \not\equiv B'$. To show that,
	let us proceed by case analysis on $A'$ assuming that $A' \equiv B'$. 
	
	Case $\unfold{A'}= \ichoice{\ell \colon A_\ell}_{\ell\in L}$.
	In this case, we would have 
	$\unfold{B'}= \ichoice{\ell \colon B_\ell}_{\ell\in L}$. Hence, 
	we would have $a_1 = \oplus \ell$ for some $\ell\in L$ and $w=w_0\cdot a_1$
	would occur in both traces and would be greater than $w_0$, which is a 
	contradiction.

	Case $\unfold{A'}= \echoice{\ell \colon A_\ell}_{\ell\in L}$. 
	Similar to the previous case.
	
	Case $\unfold{A'}= A_1 \tensor A_2$. In this case we would have
	$\unfold{B'}= B_1 \tensor B_2$. Hence, $a_1\in \{\tensor_1, \tensor_2\}$
	and we would have $w=w_0 \cdot a_1$ occurring in both traces, which contradicts 
	the assumption that $w_0$ is the greatest prefix occurring in both traces.
	
	Case $\unfold{A'}= A_1 \lolli A_2$. Similar to the 
	previous case.
	
	Case $\unfold{A'}= \one$. In this case, we would have $\unfold{B'} = \one$.
	Hence, $w_0$ would be the maximal trace for both terms, which is a 
	contradiction with the fact that $w_0$ is a prefix of a sequence of actions
	in the difference of the traces.
	
	Since all cases led to contradictions, we have $A'\not\equiv B'$. 
	The conclusion that $A \not \equiv B$, follows immediately from the property: 
	if 
	$\llparenthesis A_0\rrparenthesis \xrightarrow{w}\llparenthesis A_1\rrparenthesis$
	and $\llparenthesis B_0\rrparenthesis \xrightarrow{w}\llparenthesis B_1\rrparenthesis$
	and $A_1\not\equiv B_1$,
	then $A_0 \not\equiv B_0$. We prove this property by induction on the length of $w$.
	If $|w| = 0$, then $A_1$ coincides with $A_0$ and $B_1$
	coincides with $B_0$, so $A_0\not\equiv B_0$. Now, let $n>0$ and 
	assume the property holds
	for any trace of length $n$. Consider $w = w' \cdot a$ with $|w'| = n$ and 
	let $A_2$, $B_2$ be s.t.
	$\llparenthesis A_0\rrparenthesis \xrightarrow{w' }\llparenthesis A_2\rrparenthesis
	\xrightarrow{a}\llparenthesis A_1\rrparenthesis$
	and
	$\llparenthesis B_0\rrparenthesis \xrightarrow{w'}\llparenthesis B_2\rrparenthesis
	\xrightarrow{a}\llparenthesis B_1\rrparenthesis$.
%
	With a case analysis on $A_2$, similar to the analysis above, since $A_1\not\equiv B_1$,
	we conclude that $A_2 \not\equiv B_2$. By induction hypothesis we have $A_0\not\equiv B_o$.
	
	For the reciprocal implication, assume that 
	$\llparenthesis A\rrparenthesis \sim_{\mathcal{S}}	\llparenthesis B\rrparenthesis$.
	Consider the relation 
	\begin{equation*}
		\mathcal{R} = \{(A_0,B_0) \mid \m{trace}_{\mathcal{S}}(\llparenthesis A_0\rrparenthesis) = \m{trace}_{\mathcal{S}}(\llparenthesis B_0\rrparenthesis)\}\subseteq  \mi{Type} \times \mi{Type}.
	\end{equation*}
	Obviously, $(A,B)\in \mathcal{R}$. To prove that $\mathcal{R}$ is
	a type bisimulation, let $(A_0, B_0)\in \mathcal{R}$ and proceed 
	by case analysis on $A_0$ and $B_0$. We sketch a couple of cases for $A_0$.
	The other cases are analogous.
	
	Case $\unfold{A_0} = \ichoice{\ell \colon A_\ell}_{\ell\in L}$.
	In this case we have 
	$\llparenthesis A_0\rrparenthesis\xrightarrow{\ichoiceop \ell} \llparenthesis A_\ell\rrparenthesis$.
	Since, by hypothesis, the traces coincide,
	$\m{trace}_{\mathcal{S}}(\llparenthesis A_0\rrparenthesis) = \m{trace}_{\mathcal{S}}(\llparenthesis B_0\rrparenthesis)$,
	we have
	$\llparenthesis B_0\rrparenthesis\xrightarrow{\ichoiceop \ell} \llparenthesis B_\ell\rrparenthesis$ and,
	thus, $\unfold{B_0} = \ichoice{\ell \colon B_\ell}_{\ell\in L}$.
	Moreover, using Observation 3 of Jan\v{c}ar~\cite{Jancar10},
	we have
	$\m{trace}_{\mathcal{S}}(\llparenthesis A_\ell\rrparenthesis) = \m{trace}_{\mathcal{S}}(\llparenthesis B_\ell\rrparenthesis)$.
	Hence, $(A_\ell,B_\ell)\in \mathcal{R}$.

	Case $\unfold{A_0} = \one$. In this case,
	$\m{trace}_{\mathcal{S}}(\llparenthesis A_0\rrparenthesis) = \m{trace}_{\mathcal{S}}(\bot) = \emptyset $.
	Since $B_0$ is a closed type and 
	$\m{trace}_{\mathcal{S}}(\llparenthesis A_0\rrparenthesis) = \m{trace}_{\mathcal{S}}(\llparenthesis B_0\rrparenthesis)$
	and the types are contractive,
	we have $\unfold{B_0} = \one$.
\end{proof}

However, type equality is not only restricted to closed types
(see~\autoref{def:tpeq_vars}).
To decide equality for open types, i.e.
$\forall \vars. \, A \equiv B$ given signature $\Sg$,
we introduce a fresh label $\ell_{\alpha}$ and type
$A_{\alpha}$ for each $\alpha \in \vars$.
We extend the signature with type definitions: 
$\Sigma^* = \Sigma \cup_{\alpha \in \vars} \{A_\alpha =
\ichoice{\ell_\alpha \colon A_\alpha}\}$.
We then replace all occurrences of $\alpha$ in $A$
and $B$ with $A_{\alpha}$ and check their equality with
signature $\Sg^*$.
We prove that this substitution preserves equality.
\begin{theorem}
\label{thm:open-equiv-grammars}
	$\forall \vars. \, A \equiv_{\Sigma} B$ iff 
	$A[\sigma^*] \equiv_{\Sigma^*} B[\sigma^*]$ where 
	$\sigma^*(\alpha) = A_\alpha$
	for all $\alpha\in \vars$.
\end{theorem}

\begin{proof}
	The direct implication is immediate because $\sigma^*$ is a 
	closed substitution. For the reciprocal implication, assume that
	$\forall \vars. \, A \not\equiv_{\Sigma} B$. Either for any 
	closed substitution $\sigma: \vars$, $A[\sigma] \not\equiv B[\sigma]$,
	in which case $A[\sigma^*] \not\equiv B[\sigma^*]$; or there exists
	$\sigma':\vars$ s.t. $A[\sigma'] \not\equiv B[\sigma']$.
	In the latter, there is a distinct trace for $A[\sigma']$ and 
	$B[\sigma']$, resulting from the substitution. Thus, a maximal
	trace $w_1$ belonging to both $\m{trace}(A[\sigma'])$ and 
	$\m{trace}(B[\sigma'])$ leads to a subterm $C$ of $\sigma'(\beta)$
	and to a subterm $D$ of $\sigma'(\gamma)$, respectively:
	$A[\sigma']\xrightarrow{w_1} C$ and 
	$B[\sigma']\xrightarrow{w_1} D$, where $\beta\neq \gamma$.
	In that case, there is a \emph{subtrace} $w_0$ of $w_1$ 
	such that $A\xrightarrow{w_0} \beta$ 
	and $B\xrightarrow{w_0} \gamma$ . Hence, we conclude that 
	$A[\sigma^*]\xrightarrow{w_0} A_\beta$ and 
	$B[\sigma^*]\xrightarrow{w_0} A_\gamma$ and 
	$A_\beta \not\equiv_{\Sigma^*} A_\gamma$,
	because $\ell_\beta$ and $\ell_\gamma$ are distinct labels.
\end{proof}

\begin{theorem}
\label{thm:type-equiv-decidable}
	Checking $\forall \vars. \, A \equiv B$ is decidable.
\end{theorem}
\begin{proof}
	\autoref{thm:open-equiv-grammars} reduces equality of open types
	to equality of closed types.
	\autoref{thm:equiv-grammars} reduces equality of closed nested
	session types to trace equivalence of
	first-order grammars. Jan\v{c}ar~\cite{Jancar10} proved that
	trace equivalence for first-order grammars is decidable,
	hence establishing the decidability of equality for nested
	session types.
\end{proof}

%% file: tpeq.tex
\section{Practical Algorithm for Type Equality}\label{sec:algorithm}
Although type equality can be reduced to trace equivalence for
first-order grammars (\autoref{thm:equiv-grammars}),
the latter problem has a very high theoretical complexity
with no known practical algorithm~\cite{Jancar10}.
In response, we have designed a coinductive algorithm
for approximating type equality. Taking inspiration from Gay and
Hole~\cite{Gay2005}, we attempt to construct a bisimulation.
Our proposed algorithm is sound but incomplete and can terminate in
three states:
\emph{(i)} types are proved equal by constructing a bisimulation,
\emph{(ii)} counterexample detected by identifying the position where
types differ, or
\emph{(iii)} terminated without a conclusive answer due to incompleteness.
We interpret both \emph{(ii)} and \emph{(iii)} as a failure of
type-checking (but there is a recourse; see
Section~\ref{sec:eqtypes}).
The algorithm is
deterministic (no backtracking) and the implementation is quite
efficient in practice. For all our examples, type checking is
instantaneous (see Section~\ref{sec:impl}).

The fundamental operation in the equality algorithm is \emph{loop
detection} where we determine if we have already added
an equation $A \equiv B$ to the bisimulation we are constructing.
Due to the presence \emph{open types} with free type variables,
determining if we have considered an equation already becomes a
difficult operation.
To that purpose, we make an initial pass over the given types and introduce
fresh \emph{internal names} abstracted over their free type variables.
In the resulting signature defined type
names and type operators alternate and we can perform loop
detection entirely on defined type names (whether internal or external).


\begin{example}[Queues]\label{ex:queue}
  After creating internal names $\%i$ for the type
  $\m{queue}[\alpha] = \echoice{\mb{ins} : \alpha \lolli \m{queue}[\alpha],
  \mb{del} : \ichoice{\mb{none} : \one, \mb{some} : \alpha \tensor \m{queue}[\alpha]}}$ we obtain
  the following signature (note the alternation between type names and operators).
  \begin{sill}
    \quad $\m{queue}[\alpha] = {\with}\{\mb{ins} : \%0[\alpha], \mb{del} : \%2[\alpha]\}$ \= 
    \quad $\%0[\alpha] = \%1[\alpha] \lolli \m{queue}[\alpha]$ \\
    \quad $\%1[\alpha] = \alpha$
    \quad $\%2[\alpha] = {\oplus}\{\mb{none} : \%3, \mb{some} : \%4[\alpha]\}$
    \quad $\%3 = \one$ \\
    \quad $\%4[\alpha] = \%5[\alpha] \tensor \m{queue}[\alpha]$
    \quad $\%5[\alpha] = \alpha$
  \end{sill}
\end{example}

Based on the invariants established by internal names, the algorithm
only needs to alternately compare two type names or two \emph{structural types},
i.e., types with an operator on the head.  The
rules are shown in \autoref{fig:tpeq_rules}.  The judgment has the form
$\vars \semi \G \vdash_{\Sg} A \equiv B$ where $\vars$ contains
the free type variables in the types
$A$ and $B$, $\Sg$ is a fixed \emph{valid} signature containing type definitions
of the form $V \indv{\alpha} = C$,
and $\G$ is a collection of \emph{closures}
$\clo{\vars}{V_1 \indv{A_1} \equiv V_2\indv{A_2}}$.
If a derivation can be constructed,
all \emph{closed instances} of all closures are included in the
resulting bisimulation (see the proof of \autoref{thm:tpeq_sound}).  A
closed instance of closure $\clo{\vars}{V_1 \indv{A_1} \equiv V_2 \indv{A_2}}$
is obtained by applying a closed substitution $\sigma$ over variables
in $\vars$, i.e., $V_1 \indv{A_1[\sigma]} \equiv V_2 \indv{A_2[\sigma]}$
such that the types $V_1 \indv{A_1[\sigma]}$ and $V_2 \indv{A_2[\sigma]}$
have no free type variables.
Because the signature $\Sg$ is fixed, we elide it from the rules in
\autoref{fig:tpeq_rules}.

\begin{figure}[t]
\begin{mathpar}
  \infer[\oplus]
  {\vars \semi \G \vdash
  \ichoice{\ell : A_\ell}_{\ell \in L} \equiv \ichoice{\ell : B_\ell}_{\ell \in L}}
  {\vars \semi \G \vdash
  A_\ell \equiv B_\ell \quad (\forall \ell \in L)}
  \and
  \infer[\with]
  {\vars \semi \G \vdash
  \echoice{\ell : A_\ell}_{\ell \in L} \equiv \echoice{\ell : B_\ell}_{\ell \in L}}
  {\vars \semi \G \vdash
  A_\ell \equiv B_\ell \quad (\forall \ell \in L)}
  \and
  \infer[\tensor]
  {\vars \semi \G \vdash
  A_1 \tensor A_2 \equiv B_1 \tensor B_2}
  {\vars \semi \G \vdash A_1 \equiv B_1 \quad
  \vars \semi \G \vdash A_2 \equiv B_2}
  \and
  \infer[\lolli]
  {\vars \semi \G \vdash
  A_1 \lolli A_2 \equiv B_1 \lolli B_2}
  {\vars \semi \G \vdash A_1 \equiv B_1 \quad
  \vars \semi \G \vdash A_2 \equiv B_2}
  \and
  \infer[\one]
  {\vars \semi \G \vdash \one \equiv \one}
  {}
  \and
  \infer[\m{var}]
  {\vars \semi \G \vdash \alpha \equiv \alpha}
  {\alpha \in \vars}
  \and
  \infer[\m{refl}]
  {\vars \semi \G \vdash V\indv{A} \equiv V \indv{A'}}
  {\vars \semi \G \vdash \overline{A} \equiv \overline{A'}}
  \and 
  \inferrule*[right = $\m{expd}$]
  {V_1 \indv{\alpha_1} = A \in \Sg \and
  V_2 \indv{\alpha_2} = B \in \Sg \and
  \cons = \clo{\vars}{V_1\indv{A_1} \equiv V_2\indv{A_2}} \\\\
  \vars \semi \G, \cons
  \vdash_{\Sg} A[\overline{A_1} / \overline{\alpha_1}] \equiv
  B[\overline{A_2}/\overline{\alpha_2}]}
  {\vars \semi \G \vdash_{\Sg}
  V_1 \indv{A_1} \equiv V_2 \indv{A_2}}
  \and
  \inferrule*[right=$\m{def}$]
  {
    \clo{\vars'}{V_1 \indv{A_1'} \equiv V_2 \indv {A_2'}} \in \G \\
    \exists \sigma' : \vars'. \, \left(\vars \semi \G \vdash V_1 \indv{A_1'[\sigma']}
    \equiv V_1 \indv{A_1} \land \vars \semi \G \vdash V_2 \indv{A_2'[\sigma']} \equiv
    V_2 \indv{A_2}\right)
  }
  {
    \vars \semi \G \vdash V_1 \indv{A_1} \equiv V_2 \indv{A_2}
  }
\end{mathpar}
\vspace{-2em}
\caption{Algorithmic Rules for Type Equality}
\vspace{-2em}
\label{fig:tpeq_rules}
\end{figure}

In the type equality algorithm,
the rules for type operators simply compare the components.  If the
type constructors (or the label sets in the $\oplus$ and $\with$ rules)
do not match, then type equality fails having constructed a
counterexample to bisimulation. Similarly, two type variables are
considered equal iff they have the same name, as exemplified by the
$\m{var}$ rule. Finally, to account for $\alpha$-renaming, when
comparing explicitly quantified types (rule $\exists^\gamma$,
$\forall^\gamma$), we substitute $\alpha$ and $\beta$ by the
\emph{same fresh variable} $\gamma$.

The rule of reflexivity is needed explicitly here (but not in the
version of Gay and Hole) due to the incompleteness of the
algorithm: we may otherwise fail to recognize type names parameterized
with equal types as equal. Note that the $\m{refl}$ rule checks
a sequence of types.

Now we come to the key rules, $\m{expd}$ and $\m{def}$.  In the
$\m{expd}$ rule we expand the definitions of $V_1\indv{A_1}$ and
$V_2\indv{A_2}$, and add the closure
$\clo{\vars}{V_1\indv{A_1} \equiv V_2\indv{A_2}}$ to $\G$.
Since the equality of $V_1\indv{A_1}$ and
$V_2\indv{A_2}$ must hold for all its closed instances, the
extension of $\G$ with the corresponding closure remembers exactly that.

The $\m{def}$ rule only applies when there already exists a closure in $\G$
with the same type names $V_1$ and $V_2$. In that case, we try to find
a substitution $\sigma'$ over $\vars'$ such that $V_1 \indv{A_1}$ is equal
to $V_1 \indv{A_1'[\sigma']}$ and $V_2 \indv{A_2}$ is equal
to $V_2 \indv{A_2'[\sigma']}$. Immediately after, the $\m{refl}$
rule applies and recursively calls the equality algorithm
on both type parameters. Existence
of such a substitution ensures that any closed instance of
$\clo{\vars}{V_1 \indv{A_1} \equiv V_2 \indv{A_2}}$ is also a closed
instance of $\clo{\vars'}{V_1 \indv{A_1'} \equiv V_2 \indv{A_2'}}$,
which are already present in the constructed type bisimulation, and
we can terminate our equality check, having successfully
\emph{detected a loop}.

The algorithm so far is sound, but potentially non-terminating.
There are two points of non-termination:
\emph{(i)} when encountering name/name equations, we can use the
$\m{expd}$ rule indefinitely, and
\emph{(ii)} we call the type equality recursively
in the $\m{def}$ rule. To ensure termination in the former case, we restrict the
$\m{expd}$ rule so that for any pair of type names $V_1$ and $V_2$ there
is an upper bound on the number of closures of the form
$\clo{-}{V_1 [-] \equiv V_2 [-]}$ allowed in $\G$.
We define this upper bound as the \emph{depth bound} of the
algorithm and allow the programmer to specify this depth bound.
Surprisingly, a depth bound of 1 suffices for all of our examples.
In the latter
case, instead of calling the general type equality algorithm, we introduce
the notion of \emph{rigid equality}, denoted by $\vars \semi \G \Vdash A \equiv B$.
The only difference between general and rigid equality is that we cannot employ
the $\m{expd}$ rule for rigid equality. Since the size of the types reduce
in all equality rules except for $\m{expd}$, this algorithm terminates.
When comparing two instantiated type names, our algorithm first tries
reflexivity, then tries to close a loop with $\m{def}$, and only if
neither of these is applicable or fails do we expand the definitions
with the $\m{expd}$ rule.
Note that if type names have no parameters, our algorithm specializes to Gay and
Hole's (with the small optimizations of reflexivity and internal
naming), which means our algorithm is sound and complete on monomorphic
types.

\paragraph{\textbf{Soundness.}}
We establish the soundness of the equality algorithm by constructing
a type bisimulation from a derivation of $\vars \semi \G \vdash A \equiv B$ by
\emph{(i)} collecting the conclusions of all the sequents,
and \emph{(ii)} forming all closed instances from them.

\begin{definition}
  Given a derivation $\mathcal{D}$ of
  $\vars \semi \G \vdash A \equiv B$, we define the set
  $\mathcal{S}(\mathcal{D})$ of closures.  For each sequent
  (regular or rigid) of the form
  $\vars \semi \G \vdash A \equiv B$ in $\mathcal{D}$,
  we include the closure
  $\clo{\vars}{A \equiv B}$ in $\mathcal{S}(\mathcal{D})$.
\end{definition}

\begin{lemma}[Closure Invariants]\label{lem:closure}
  For any valid derivation $\mathcal{D}$ with the set of closures $\SD$,
  \begin{itemize}
    \item If $\clo{\vars}{\ichoice{\ell : A_\ell}_{\ell \in L} \equiv
    \ichoice{\ell : B_\ell}_{\ell \in L}} \in \SD$ from $\ichoiceop$ rule, then
    $\clo{\vars}{A_\ell \equiv B_\ell} \in \SD$ for all $\ell \in L$.

    \item If $\clo{\vars}{\echoice{\ell : A_\ell}_{\ell \in L} \equiv
    \echoice{\ell : B_\ell}_{\ell \in L}} \in \SD$ from $\echoiceop$ rule, then
    $\clo{\vars}{A_\ell \equiv B_\ell} \in \SD$ for all $\ell \in L$.

    \item If $\clo{\vars}{A_1 \tensor A_2 \equiv
    B_1 \tensor B_2} \in \SD$ from $\tensor$ rule, then
    $\clo{\vars}{A_1 \equiv B_1} \in \SD$ and
    $\clo{\vars}{A_2 \equiv B_2} \in \SD$.

    \item If $\clo{\vars}{A_1 \lolli A_2 \equiv
    B_1 \lolli B_2} \in \SD$ from $\lolli$ rule, then
    $\clo{\vars}{A_1 \equiv B_1} \in \SD$ and
    $\clo{\vars}{A_2 \equiv B_2} \in \SD$.

    \item If $\clo{\vars}{V \indv{A_1} \equiv V \indv{A_2}} \in \SD$
    from $\m{refl}$ rule, then for each $\clo{\vars}{A_1^i \equiv A_2^i}
    \in \SD$ for each $i$ in $1 .. |\overline{A}|$.

    \item If $\clo{\vars}{V_1 \indv{A_1} \equiv V_2 \indv{A_2}} \in \SD$
    from $\m{expd}$ rule and $V_1 \indv{\alpha_1} = B_1 \in \Sg$ and
    $V_2 \indv{\alpha_2} = B_2 \in \Sg$, then
    $\clo{\vars}{B_1[\overline{A_1}/\alpha_1] \equiv
    B_2[\overline{A_2}/\alpha_2]} \in \SD$.
  \end{itemize}
\end{lemma}

\begin{proof}
  By induction on the type equality judgment.
\end{proof}

\begin{theorem}[Soundness]\label{thm:tpeq_sound}
  If $\vars \semi \cdot \vdash A \equiv B$, then
  $\forall \vars.\, A \equiv B$. Consequently,
  if $\vars$ is empty, we get $A \equiv B$.
\end{theorem}

\begin{proof}
  Given a derivation $\mathcal{D}_0$ of $\vars_0 \semi \cdot \vdash A_0 \equiv B_0$,
  construct $\mathcal{S}(\mathcal{D}_0)$ and define relation $\rel_0$ as follows:
  \[
    \rel_0 = \{(A[\sigma], B[\sigma]) \mid \clo{\vars}{A \equiv B}
    \in \mathcal{S}(\mathcal{D}_0)\;\mbox{and}\;\sigma\;\mbox{over}\;\vars \label{bis:1}\}
  \]
  Then, construct $\rel_1$ as follows:
  \[
    \rel_1 = \{(V\indv{A}, V\indv{B}) \mid V\indv{\alpha} = C \in \Sg \text{ and }
    (A^i, B^i) \in \rel_0 \; \forall i \in 1..|A|\}
  \]
  Consider $\rel$ to be the \emph{reflexive transitive closure} of $\rel_0 \cup \rel_1$.
  Note that extending a relation by its reflexive transitive closure
  preserves its bisimulation properties since the bisimulation is strong.
  We prove that $\rel$ is a type bisimulation. Then our theorem follows
  since the closure
  $\clo{\vars_0}{A_0 \equiv B_0} \in \mathcal{S}(\mathcal{D}_0)$,
  and hence, for any closed substitution $\sigma$,
  $(A_0[\sigma], B_0[\sigma]) \in \rel$.

  To prove $\rel$ is a bisimulation, we
  consider $(A[\sigma], B[\sigma]) \in \rel$ where
  $\clo{\vars}{A \equiv B} \in \mathcal{S}(\mathcal{D}_0)$ for
  some $\sigma$ over $\vars$. We case analyze on the rule in the
  derivation which added the above closure to $\rel$.

  Consider the case where $\ichoiceop$ rule is applied.
  The rule dictates that $A = \ichoice{\ell : A_\ell}_{\ell \in L}$
  and $B = \ichoice{\ell : B_\ell}_{\ell \in L}$. Since
  $\clo{\vars}{A \equiv B} \in \mathcal{S}(\mathcal{D}_0)$, by Lemma~\ref{lem:closure},
  we obtain
  $\clo{\vars}{A_\ell \equiv B_\ell} \in \mathcal{S}(\mathcal{D}_0)$ for all
  $\ell \in L$. By the definition of $\rel$, we get that
  $(A_\ell[\sigma], B_\ell[\sigma]) \in \rel$. Also,
  $A[\sigma] = \ichoice{\ell : A_\ell[\sigma]}_{\ell \in L}$ and similarly,
  $B[\sigma] = \ichoice{\ell : B_\ell[\sigma]}_{\ell \in L}$.  Hence, $\rel$ satisfies the
  closure condition (case 1) from Definition~\ref{def:rel}.
  The cases for $\with$, $\tensor$, $\lolli$ and $\one$ are analogous.

  If the applied rule is $\m{var}$, then $A = \alpha$ and $B = \alpha$.
  In this case, the relation $\rel$ contains any $(\sigma, \sigma)$ for
  a ground session type $\sigma$.
  To prove $\rel$ is a type bisimulation, we need to subcase on the form of $\sigma$.
  For instance, if $\sigma$ is of the form $\ichoice{\ell : A_\ell}$,
  then we need to prove that $(A_\ell, A_\ell) \in \rel$.
  But since $A_\ell$ is a ground session type, the definition of $\rel$
  implies that it contains $(A_\ell, A_\ell)$.
  The remaining subcases are analogous.



  Consider the case where $\m{expd}$ rule is applied. In this case,
  $A = V_1 \indv{A_1}$ and $B = V_2 \indv{A_2}$ and $(A[\sigma], B[\sigma]) \in \rel$.
  Suppose the definitions are $V_1 \indv{\alpha_1} = B_1$ and $V_2 \indv{\alpha_2} = B_2$.
  Next, we havea $\unfold{A} = B_1[\overline{A_1}/\overline{\alpha_1}]$ and
  $\unfold{B} = B_2[\overline{A_2}/\overline{\alpha_2}]$. From Lemma~\ref{lem:closure},
  we conclude that $\clo{\vars}{B_1[\overline{A_1}/\overline{\alpha_1}] \equiv
  B_2[\overline{A_2}/\overline{\alpha_2}]} \in \SDO$.
  Since the next applied rule has to be one of $\ichoiceop, \echoiceop, \tensor, \lolli,
  \one, \m{var}$, we can use Lemma~\ref{lem:closure} again to obtain the closure conditions
  of a type bisimulation.

  The most crucial case is when the applied rule is $\m{def}$ as we attempt to
  close off the loop here. In this case,
  the second premise of the $\m{def}$ rule ensures that there exists a substituion
  $\sigma'$ over $\vars'$ which entails $\clo{\vars}{V_1 \indv{A_1'[\sigma']}, V_1 \indv{A_1}}
  \in \mathcal{S}({\mathcal{D}_0})$ and $\clo{\vars}{V_2 \indv{A_2'[\sigma']}, V_2 \indv{A_2}}
  \in \mathcal{S}({\mathcal{D}_0})$.
  To satisfy the closure condition, we need to prove
  $(V_1 \indv{A_1 [\sigma]}, V_2 \indv{A_2 [\sigma]}) \in \rel$ for any closed
  substitution $\sigma$ over $\vars$.
  The key lies in composing the closed substitution $\sigma$ with the substitution
  $\sigma' : \vars' \to \vars$ to obtain the closed substitution $\sigma \circ \sigma'$
  over $\vars'$.
  The closure $\clo{\vars}{V_1 \indv{A_1'[\sigma']}, V_1 \indv{A_1}}
  \in \mathcal{S}({\mathcal{D}_0})$ implies
  that $(V_1 \indv{A_1'[\sigma \circ \sigma']}, V_1 \indv{A_1 [\sigma]}) \in \rel$.
  The closure $\clo{\vars}{V_2 \indv{A_2'[\sigma']}, V_2 \indv{A_2}}
  \in \mathcal{S}({\mathcal{D}_0})$ implies
  $(V_2 \indv{A_2'[\sigma \circ \sigma']}, V_2 \indv{A_2 [\sigma]}) \in \rel$.
  The first premise states that
  $\clo{\vars'}{V_1 \indv{A_1'} \equiv
  V_2 \indv{A_2'}} \in \mathcal{S}(\mathcal{D}_0)$.
  This entails that $(V_1 \indv{A_1'[\sigma_0]},
  V_2 \indv{A_2'[\sigma_0]}) \in \rel$ for any substituion $\sigma_0$ over $\vars'$.
  Setting $\sigma_0 = \sigma \circ \sigma'$ entails
  $(V_1 \indv{A_1'[\sigma \circ \sigma']}, V_2 \indv{A_2'[\sigma \circ \sigma']}) \in \rel$.
  The transitive property of $\rel$ then ensures that
  $(V_1 \indv{A_1}, V_2 \indv{A_2}) \in \rel$.

  The last two cases concern reflexivity, one that comes from directly
  the closure obtained from applying the $\m{refl}$ rule, and the other
  comes from the relation $\rel_1$.
  First, consider the case $(V \indv{A_1[\sigma]}, V \indv{A_2[\sigma]}) \in \rel$ which is
  added due to the closure $\clo{\vars}{V \indv{A_1} \equiv V \indv{A_2}} \in \SDO$.
  Lemma~\ref{lem:closure} ensures that $\clo{\vars}{{A_1}^i \equiv {A_2}^i} \in \SDO$
  for each $i$, implying $({A_1}^i[\sigma], {A_2}^i[\sigma]) \in \rel_0$.
  And suppose $V \indv{\alpha} = B \in \Sg$. We subcase on the form of $B$.
  Consider the representative subcase where $B = \ichoice{\ell : B_\ell}_{\ell \in L}$.
  To prove $\rel$ is a bisimulation, we need to prove
  $(B_\ell[\overline{A_1[\sigma]}/\overline{\alpha}], B_\ell[\overline{A_2[\sigma]}/\overline{\alpha}])
  \in \rel$. Note however the internal renaming condition ensures that
  $B_\ell = V_B \indv{\alpha}$ for some (possibly internal) type name $V_B$.
  But then the definition of $\rel_1$ coupled with the consequence of
  Lemma~\ref{lem:closure} ensures that $(V_B \indv{A_1[\sigma]}, V_B \indv{A_2[\sigma]})
  \in \rel$, satisfying the closure condition.
  The other structural subcases are analogous.
  Consider the subcase where $B = \alpha$. Thus, $V \indv{A_1} = A_1^i$ and
  $V \indv{A_2} = A_2^i$, and since $({A_1}^i[\sigma], {A_2}^i[\sigma]) \in \rel_0$,
  they are contained in $\rel$ as well.
  
  A similar argument covers the latter case where $(V \indv{A_1[\sigma]}, V \indv{A_2[\sigma]})
  \in \rel$ due to $\rel_1$.

  Thus, $\rel$ is a bisimulation.
\end{proof}

\subsection{Type Equality Declarations}
\label{sec:eqtypes}

One of the primary sources of incompleteness in our algorithm is its
inability to \emph{generalize the coinductive hypothesis}.
As an illustration, consider the following two types $D$ and $D'$,
which only differ in the names, but have the same structure.
\begin{mathpar}
  \vspace{-0.5em}
  T[x] \triangleq \ichoice{\mb{L} : T[T[x]], \mb{R} : x} \and
  D \triangleq \ichoice{\mb{L} : T[D], \$ : \one} \\
  T'[x] \triangleq \ichoice{\mb{L} : T'[T'[x]], \mb{R} : x} \and
  D' \triangleq \ichoice{\mb{L} : T'[D'], \$ : \one}
\end{mathpar}
To establish $D \equiv D'$, our algorithm explores the $\mb{L}$
branch and checks $T[D] \equiv T'[D']$.
A corresponding closure $\clo{\cdot}{T[D] \equiv T'[D']}$ is
added to $\Gamma$, and our algorithm then checks $T[T[D]] \equiv
T'[T'[D']]$.
This process repeats until it exceeds the depth bound and terminates
with an inconclusive answer.
What the algorithm never realizes is that $T[x] \equiv T'[x]$
for all $x \in \mi{Type}$; it fails to generalize to this hypothesis and
is always inserting closed equality constraints to $\G$.

To allow a recourse, we permit the programmer to declare (concrete syntax)
\begin{Verbatim}
  eqtype T[x] = T'[x]
\end{Verbatim}
an equality constraint easily verified by our algorithm.
We then \emph{seed} the $\G$ in the equality algorithm with the corresponding
closure from the $\m{eqtype}$ constraint which can then
be used to establish $D \equiv D'$
\[
  \cdot \semi \clo{x}{T[x] \equiv T'[x]} \vdash D \equiv D'
\]
which, upon exploring the $\mb{L}$ branch reduces to
\[
  \cdot \semi \clo{x}{T[x] \equiv T'[x]}, \clo{\cdot}{D \equiv D'}
  \vdash T[D] \equiv T'[D']
\]
which holds because under the substitution $[D/x]$ as required
by the $\m{def}$ rule.

In the implementation, we first collect all the $\m{eqtype}$
declarations in the program into a global set of closures $\G_0$.
We then validate every $\m{eqtype}$ declaration by checking
$\vars \semi \G_0 \vdash A \equiv B$ for every pair $(A,B)$
(with free variables $\vars$) in the $\m{eqtype}$ declarations.
Essentially, this ensures that all equality declarations are
valid w.r.t. each other.
Finally, all equality checks are then performed under this more
general $\G_0$. The soundness of this approach can be proved
with the following more general theorem.
\begin{theorem}[Seeded Soundness]\label{thm:tpeq_sound}
  For a valid set of $\m{eqtype}$ declarations $\G_0$, if
  $\vars \semi \G_0 \vdash A \equiv B$, then
  $\forall \vars.\, A \equiv B$.
\end{theorem}

Our soundness proof can easily be modified to accommodate
this requirement. Intuitively, since $\G_0$ is valid, all
closed instances of $\G_0$ are already proven to be bisimilar.
Thus, all properties of a type bisimulation are still preserved
if all closed instances of $\G_0$ are added to it.

One final note on the rule of reflexivity: a type name may \emph{not}
actually depend on its parameter.  As a simple example, we have
$V[\alpha] = \one$; a more complicated one would be
$V[\alpha] = \ichoice{a : V[V[\alpha]], b : \one}$.  When applying
reflexivity, we would like to conclude that $V[A] \equiv V[B]$
regardless of $A$ and $B$.  This could be easily established with an
equality type declaration $\m{eqtype}\; V[\alpha] = V[\beta]$.  In
order to avoid this syntactic overhead for the programmer, we
determine for each parameter $\alpha$ of each type name $V$ whether
its definition is nonvariant in $\alpha$.  This information is
recorded in the signature and used when applying the reflexivity rule
by ignoring nonvariant arguments.

%% file: language.tex
\section{Formal Language Description}\label{sec:language}

\input{tab_desc}

In this section, we present the program constructs we have designed
to realize nested polymorphism which have also been integrated with the
Rast language~\cite{Das20FSCD,Das20CONCUR,Das20PPDP} to
support general-purpose programming.
The underlying base system of session types is derived from a Curry-Howard
interpretation~\cite{Caires10concur,Caires16mscs} of intuitionistic linear logic
\cite{Girard87tapsoft}. The key idea is that an intuitionistic linear sequent
$A_1 \; A_2 \; \ldots \; A_n \vdash A$
is interpreted as the interface to a process $P$. We label each of the
antecedents with a channel name $x_i$ and the succedent with channel name $z$.
The $x_i$'s are \emph{channels used by} $P$ and $z$ is the \emph{channel provided by} $P$.
\begin{mathpar}
  (x_1 : A_1) \; (x_2 : A_2) \ldots (x_n : A_n) \vdash P :: (z : C)
\end{mathpar}
The resulting judgment formally states that process $P$ provides a service of
session type $C$ along channel $z$, while using the services of session types $A_1,
\ldots, A_n$ provided along channels $x_1, \ldots, x_n$ respectively. All these
channels must be distinct. We abbreviate the antecedent of the sequent by $\Delta$.

Due to the presence of type variables, the formal typing judgment
is extended with $\vars$ and written as
\begin{center}
  \begin{minipage}{0cm}
  \begin{tabbing}
  $\vars \semi \D \vdash_\Sg P :: (x : A)$
  \end{tabbing}
  \end{minipage}
\end{center}
where $\vars$ stores the type variables $\alpha$,
$\D$ represents the linear antecedents $x_i : A_i$, $P$ is the process
expression and $x : A$ is the linear succedent. We propose and maintain
that all free type variables in $\D, P$, and $A$ are contained in $\vars$.
Finally, $\Sigma$ is a fixed valid signature containing type and process
definitions.
Table~\ref{tab:language} overviews the session types, their associated
process terms, their continuation (both in types and terms) and operational description.
For each type, the first line describes the provider's viewpoint, while
the second line describes the client's matching but dual viewpoint.

We formalize the operational semantics as a system of \emph{multiset rewriting
rules}~\cite{Cervesato09SEM}. We introduce semantic objects $\proc{c}{P}$
and $\msg{c}{M}$ which mean that process $P$ or message $M$ provide
along channel $c$.
A process configuration is a multiset of such objects, where any two
provided channels are distinct.

\subsection{Basic Session Types}\label{subsec:base}

We briefly review the structural types already existing in the Rast language,
focusing on explicit quantifier operators that we introduce.

\paragraph{\textbf{Structural Types}}
The \emph{internal choice} type constructor
$\ichoice{\ell : A_{\ell}}_{\ell \in L}$ is an $n$-ary labeled
generalization of the additive disjunction $A \ichoiceop B$.
Operationally, it requires the provider of
$x : \ichoice{\ell : A_{\ell}}_{\ell \in L}$ to send a label
label $k \in L$ on channel $x$ and continue to provide
type $A_{k}$. The corresponding process term is written as $(\esendl{x}{k} \semi P)$
where the continuation $P$ provides type $x : A_k$.
Dually, the client must branch based
on the label received on $x$ using the process term
$\ecase{x}{\ell}{Q_\ell}_{\ell \in L}$ where $Q_\ell$ is the continuation
in the $\ell$-th branch.
\begin{mathpar}
  \infer[{\oplus}R]
    {\vars \semi \D \vdash (\esendl{x}{k} \semi P) :: (x : \ichoice{\ell : A_\ell}_{\ell \in L})}
    {(k \in L) & \vars \semi \D \vdash P :: (x : A_k)}
  \and
  \infer[{\oplus}L]
    {\vars \semi \D, (x : \ichoice{\ell : A_\ell}_{\ell \in L}) \vdash
    \ecase{x}{\ell}{Q_\ell}_{\ell \in L} :: (z : C)}
    {(\forall \ell \in L) &
      \vars \semi \D, (x : A_\ell) \vdash Q_\ell :: (z : C)}
\end{mathpar}

Communication is asynchronous, so that the client
($\esendl{c}{k} \semi Q$) sends a message $k$ along $c$ and continues as $Q$
without waiting for it to be received. As a technical device to ensure that
consecutive messages on a channel arrive in order, the sender also creates a
fresh continuation channel $c'$ so that the message $k$ is actually represented
as $(\esendl{c}{k} \semi \fwd{c}{c'})$ (read: send $k$ along $c$ and continue along
$c'$). When the message $k$ is received along $c$, we select branch $k$ and
also substitute the continuation channel $c'$ for $c$.
\begin{tabbing}
$({\oplus}S) : \m{proc}(c, c.k \semi P) \;\mapsto\;
\m{proc}(c', P[c'/c]), \m{msg}(c, c.k \semi \fwd{c}{c'})$ \\
$({\oplus}C):$ \= $\m{msg}(c, c.k \semi \fwd{c}{c'}),
\m{proc}(d, \m{case}\;c\;(\ell \Rightarrow Q_\ell)_{\ell \in L})
\;\mapsto\; \m{proc}(d, Q_k[c'/c])$
\end{tabbing}

The \emph{external choice} constructor $\echoice{\ell : A_{\ell}}_{\ell \in L}$
generalizes additive conjunction and is the \emph{dual} of internal
choice reversing the role of the provider and client. Thus, the provider
branches on the label $k \in L$ sent by the client.
\begin{mathpar}
  \infer[{\with}R]
  {\vars \semi \D \vdash \ecase{x}{\ell}{P_\ell}_{\ell \in L} ::
  (x : \echoice{\ell : A_\ell}_{\ell \in L})}
  {(\forall \ell \in L)
   & \vars \semi \D \vdash P_\ell :: (x : A_\ell)}
  \and
  \infer[{\with}L]
  {\vars \semi \D, (x : \echoice{\ell : A_\ell}_{\ell \in L}) \vdash
  (\esendl{x}{k} \semi Q) :: (z : C)}
  {(k \in L) & \vars \semi \D, (x : A_k) \vdash Q :: (z : C)}
\end{mathpar}
Rules ${\with}S$ and ${\with}C$ below describe the operational behavior of the
provider and client respectively $\fresh{c'}$.
\begin{tabbing}
$({\with}S) : \m{proc}(d, c.k \semi Q) \;\mapsto\; \m{msg}(c', c.k \semi \fwd{c'}{c}),
\m{proc}(d, Q[c'/c])$ \\
$({\with}C):$ \= $\m{proc}(c, \m{case}\;c\;(\ell \Rightarrow P_\ell)_{\ell \in L}),
\m{msg}(c', c.k \semi \fwd{c'}{c})
\;\mapsto\; \m{proc}(c', P_k[c'/c])$
\end{tabbing}

The \emph{tensor} operator $A \tensor B$ prescribes that the provider of
$x : A \tensor B$
sends a channel, say $w$ of type $A$ and continues to provide type $B$. The
corresponding process term is $\esendch{x}{w} \semi P$ where $P$ is
the continuation.  Correspondingly, its client must receive a channel on $x$
using the term $\erecvch{x}{y} \semi Q$, binding it to variable $y$
and continuing to execute $Q$.
\begin{mathpar}
  \infer[{\tensor}R]
    {\vars \semi \D, (y : A) \vdash (\esendch{x}{y} \semi P) :: (x : A \tensor B)}
    {\vars \semi \D \vdash P :: (x : B)}
  \and
  \infer[{\tensor}L]
    {\vars \semi \D, (x : A \tensor B) \vdash (\erecvch{x}{y} \semi Q) :: (z : C)}
    {\vars \semi \D, (y : A), (x : B) \vdash Q :: (z : C)}
\end{mathpar}
Operationally, the provider $(\esendch{c}{d} \semi P)$ sends the
channel $d$ and the continuation channel $c'$ along $c$ as a message and
continues with executing $P$. The client receives the channel $d$ and continuation
channel $c'$ and substitutes $d$ for $x$ and $c'$ for $c$.
\begin{tabbing}
$({\tensor}S) : \m{proc}(c, \m{send}\; c\; d \semi P) \;\mapsto\; $\=
$\m{proc}(c', P[c'/c]),
\m{msg}(c, \m{send}\; c\; d \semi \fwd{c}{c'})$ \\
$({\tensor}C) : $ \= $\m{msg}(c, \m{send}\; c\; d \semi \fwd{c}{c'}),
\m{proc}(e, x \leftarrow \m{recv}\; c \semi Q)
\;\mapsto\; \m{proc}(e, Q[c', d/c, x])$
\end{tabbing}

The dual operator $A \lolli B$ allows the provider to receive a
channel of type $A$ and continue to provide type $B$. The client
of $A \lolli B$, on the other hand, sends the channel of type $A$
and continues to use $B$ using dual process terms as $\tensor$.
\begin{mathpar}
  \infer[{\lolli}R]
    {\vars \semi \D \vdash (\erecvch{x}{y} \semi P) :: (x : A \lolli B)}
    {\vars \semi \D, (y : A) \vdash P :: (x : B)}
  \and
  \infer[{\lolli}L]
    {\vars \semi \D, (x : A \lolli B), (y : A) \vdash (\esendch{x}{y} \semi Q) :: (z : C)}
    {\vars \semi \D, (x : B) \vdash Q :: (z : C)}
\end{mathpar}
\begin{tabbing}
$({\lolli}S) : \m{proc}(e, \m{send}\; c\; d \semi Q) \;\mapsto\; $\=
$\m{msg}(c', \m{send}\; c\; d \semi \fwd{c'}{c}),
\m{proc}(e, Q[c'/c])$ \\
$({\lolli}C) : $ \= $\m{proc}(c, x \leftarrow \m{recv}\; c \semi P),
\m{msg}(c', \m{send}\; c\; d \semi \fwd{c'}{c})
\;\mapsto\; \m{proc}(c', P[c', d/c, x])$
\end{tabbing}

The type $\one$
indicates \emph{termination} requiring that the provider of $x : \one$
send a \emph{close} message, formally written as $\eclose{x}$
followed by terminating the communication. Correspondingly,
the client of $x : \one$ uses the term $\ewait{x} \semi Q$
to wait for $x$ to terminate before continuing with executing $Q$.
Linearity enforces that the provider does not use any channels.
\begin{mathpar}
  \infer[{\one}R]
    {\vars \semi \cdot \vdash (\eclose{x}) :: (x : \one)}
    {}
  \and
  \infer[{\one}L]
    {\vars \semi \D, (x : \one) \vdash (\ewait{x} \semi Q) :: (z : C)}
    {\vars \semi \D \vdash Q :: (z : C)}
\end{mathpar}
Operationally, the provider waits for the closing message, which
has no continuation channel since the provider terminates.
\begin{tabbing}
$({\one}S) : \m{proc}(c, \m{close}\; c) \;\mapsto\; \m{msg}(c, \m{close}\; c)$ \\
$({\one}C) : $ \= $\m{msg}(c, \m{close}\; c),
\m{proc}(d, \m{wait}\; c \semi Q) \;\mapsto\; \m{proc}(d, Q)$
\end{tabbing}

A forwarding process $\fwd{x}{y}$ identifies the channels $x$ and $y$ so that any
further communication along either $x$ or $y$ will be along the unified channel.
Its typing rule corresponds to the logical rule of identity.
\begin{mathpar}
  \infer[\m{id}]
    {\vars \semi y : A \vdash (\fwd{x}{y}) :: (x : A)}
    {}
\end{mathpar}
Operationally, a process $\fwd{c}{d}$ \emph{forwards} any message M
that arrives on $d$ to $c$ and vice-versa. Since channels are used
linearly, the forwarding process can then terminate, ensuring proper
renaming, as exemplified in the rules below.
\begin{tabbing}
$(\m{id}^+C) : $ \= $\m{msg}(d', M),
\m{proc}(c, \fwd{c}{d}) \;\mapsto\; \m{msg}(c, M[c/d])$ \\
$(\m{id}^-C) : $ \= $\m{proc}(c, \fwd{c}{d}),
\m{msg}(e, M(c)) \;\mapsto\; \m{msg}(e, M(c)[d/c])$
\end{tabbing}
We write $M(c)$ to indicate that $c$ must occur in message $M$ ensuring that $M$ is the
sole client of $c$.

\paragraph{\textbf{Process Definitions}}
Process definitions have the form
$\D \vdash f\indv{\alpha} = P :: (x : A)$ where $f$ is the name of the
process and $P$ its definition, with $\D$ being the channels used
by $f$ and $x : A$ being the offered channel. In addition, $\overline{\alpha}$
is a sequence of type variables that $\D$, $P$ and $A$ can refer to.
All definitions are collected in the fixed global signature $\Sg$.
For a \emph{valid signature}, we
require that $\overline{\alpha} \semi \D \vdash P :: (x : A)$
for every definition, thereby allowing
definitions to be mutually recursive. A new instance of a defined
process $f$ can be spawned with the expression
$\ecut{x}{f \indv{A}}{\overline{y}}{Q}$ where $\overline{y}$ is a
sequence of channels matching the antecedents $\D$
and $\overline{A}$ is a sequence of types matching the type
variables $\overline{\alpha}$.
The newly spawned process will use all variables in
$\overline{y}$ and provide $x$ to the continuation $Q$.
\vspace{-0.5em}
\begin{mathpar}
  \inferrule*[right=$\m{def}$]
  {\overline{y':B'} \vdash f \indv{\alpha} = P_f :: (x' : B) \in \Sg \quad
  \D' = \overline{(y:B')}[\overline{A}/\overline{\alpha}] \\\\
  \vars \semi \D, (x : B[\overline{A}/\overline{\alpha}]) \vdash Q :: (z : C)}
  {\vars \semi \D, \D' \vdash (\ecut{x}{f \indv{A}}{\overline{y}}{Q}) :: (z : C)}
\end{mathpar}
The declaration of $f$ is looked up in the signature $\Sg$ (first premise),
and $\overline{A}$ is substituted for $\overline{\alpha}$ while
matching the types in $\D'$ and $\overline{y}$ (second premise). Similarly,
the freshly created channel $x$ has type $A$ from the signature
with $\overline{A}$ substituted for $\overline{\alpha}$.
The corresponding semantics rule also performs a similar substitution
$\fresh{a}$.
\begin{tabbing}
$(\m{def}C) : $ \= $\m{proc}(c, \ecut{x}{f \indv{A}}{\overline{d}}{Q}) \; \mapsto \;
\m{proc}(a, P_f[a/x, \overline{d}/\overline{y'}, \overline{A}/\overline{\alpha}]), \;
   \m{proc}(c, Q[a/x])$
\end{tabbing}
where $\overline{y' : B'} \vdash f = P_f :: (x' : B) \in \Sg$.

Sometimes a process invocation is a tail call,
written without a continuation as $\procdef{f \indv{A}}{\overline{y}}{x}$. This is a
short-hand for $\procdef{f \indv{A}}{\overline{y}}{x'} \semi \fwd{x}{x'}$ for a fresh
variable $x'$, that is, we create a fresh channel
and immediately identify it with x.




\subsection{Type Safety}\label{subsec:soundness}
The extension of session types with nested polymorphism
is proved type safe by the standard theorems of
\emph{preservation} and \emph{progress}, also known as
\emph{session fidelity} and \emph{deadlock freedom}.
At runtime, a program is represented using a multiset of semantic
objects denoting processes and messages defined as a \emph{configuration}.
\begin{center}
  \begin{minipage}{0cm}
  \begin{tabbing}
  $\config \; ::= \; \cdot \mid \config, \config' \mid \proc{c}{P} \mid \msg{c}{M}$
  \end{tabbing}
  \end{minipage}
\end{center}
We say that $\proc{c}{P}$ (or $\msg{c}{M}$) provide channel $c$.
We stipulate that no two distinct semantic objects in a configuration
provide the same channel.

\paragraph{\textbf{Type Preservation}}
The key to preservation is defining the rules to
\emph{type a configuration}.
We define a well-typed configuration
using the judgment $\D_1 \Vdash_{\Sg} \config :: \D_2$
denoting that configuration $\config$ uses channels $\D_1$
and provides channels $\D_2$. A configuration is always
typed w.r.t. a valid signature $\Sg$.
Since the signature $\Sg$ is fixed, we elide it from
the presentation.

\begin{figure}[t]
  \begin{mathpar}
  \infer[\m{emp}]
  {\D \Vdash (\cdot) :: \D}
  {}
  \and
  \infer[\m{comp}]
  {\D_1 \Vdash (\config_1, \config_2) :: \D_3}
  {\D_1 \Vdash \config_1 :: \D_2 \and
  \D_2 \Vdash \config_2 :: \D_3}
  \and
  \infer[\m{proc}]
  {\D \Vdash \proc{x}{P} :: (x:A)}
  {\cdot \semi \D \vdash P :: (x : A)}
  \and
  \infer[\m{msg}]
  {\D \Vdash \msg{x}{M} :: (x:A)}
  {\cdot\semi \D \vdash M :: (x : A)}
  \end{mathpar}
  \vspace{-1em}
  \caption{Typing rules for a configuration}
  \vspace{-1em}
  \label{fig:config_typing}
\end{figure}

The rules for typing a configuration are defined in
Figure~\ref{fig:config_typing}.
The $\m{emp}$ rule states that an empty configuration
does not consume any channels provides all channels it uses.
The $\m{comp}$
rule composes two configurations $\config_1$ and $\config_2$;
$\config_1$ provides channels $\D_2$ while $\config_2$ uses
channels $\D_2$. The rule $\m{proc}$ creates a singleton configuration
out of a process. Since configurations are runtime objects,
they do not refer to any free variables and $\vars$ is empty.
The $\m{msg}$ rule is analogous.

\paragraph{\textbf{Global Progress}}
To state progress, we need to define a \emph{poised
process}~\cite{Pfenning15fossacs}.  A process $\proc{c}{P}$ is
poised if it is trying to receive a message on $c$. Dually, a message
$\msg{c}{M}$ is poised if it is sending along $c$. A configuration
is poised if every message or process in the configuration is
poised. Intuitively, this represents that the
configuration is trying to communicate \emph{externally} along
one of the channels it uses or provides.
\begin{theorem}[Type Safety]\label{thm:type_safety}
  For a well-typed configuration $\D_1 \Vdash_{\Sg} \config :: \D_2$,
  \begin{enumerate}
    \item[(i)] (Preservation) If $\config \step \config'$, then
    $\D_1 \Vdash_{\Sg} \config' :: \D_2$

    \item[(ii)] (Progress) Either $\config$ is poised, or $\config \step
    \config'$. 
  \end{enumerate}
\end{theorem}

\begin{proof}
  Preservation is proved by case analysis on the rules of
  operational semantics. First, we invert the derivation of
  the current configuration $\config$ and use the premises to
  assemble a new derivation for $\config'$.
  Progress is proved by induction on the right-to-left typing of $\config$
  so that either $\config$ is empty (and therefore poised) or $\config =
  (\dc, \proc{c}{P})$ or $\config = (\dc, \msg{c}{M})$. By the induction
  hypothesis, either $\dc \step \dc'$ or $\dc$ is poised.
  In the former case, $\config$ takes a step (since $\dc$ does).
  In the latter case, we analyze the cases for $P$ and $M$,
  applying multiple steps of inversion to show that in each case either
  $\config$ can take a step or is poised.
\end{proof}

%% file: tab_desc.tex
\begin{table*}[t]
  \centering
  \small
  \begin{tabular}{l l l l l}
  \textbf{Type} & \textbf{Cont.} & \textbf{Process Term} & \textbf{Cont.} & \multicolumn{1}{c}{\textbf{Description}} \\
  \toprule
  $c : \ichoice{\ell : A_\ell}_{\ell \in L}$ & $c : A_k$ & $\esendl{c}{k} \semi P$
  & $P$ & send label $k$ on $c$ \\
  & & $\ecase{c}{\ell}{Q_\ell}_{\ell \in L}$ & $Q_k$ & receive label $k$ on $c$ \\
  \addlinespace
  $c : \echoice{\ell : A_\ell}_{\ell \in L}$ & $c : A_k$ & $\ecase{c}{\ell}{P_\ell}_{\ell \in L}$
  & $P_k$ & receive label $k$ on $c$ \\
  & & $\esendl{c}{k} \semi Q$ & $Q$ & send label $k$ on $c$ \\
  \addlinespace
  $c : A \tensor B$ & $c : B$ & $\esendch{c}{w} \semi P$
  & $P$ & send channel $w : A$ on $c$ \\
  & & $\erecvch{c}{y} \semi Q_y$ & $Q_y[w/y]$ & receive channel $w : A$ on $c$ \\
  \addlinespace
  $c : A \lolli B$ & $c : B$ & $\erecvch{c}{y} \semi P_y$
  & $P_y[w/y]$ & receive channel $w : A$ on $c$ \\
  & & $\esendch{c}{w} \semi Q$ & $Q$ & send channel $w : A$ on $c$ \\
  \addlinespace
  $c : \texists{\alpha}A $ & $c : A[B/\alpha]$ & $\esendch{c}{[B]} \semi P$
  & $P$ & send type $B$ on $c$ \\
  & & $\erecvch{c}{[\alpha]} \semi Q_{\alpha}$ & $Q_{\alpha}[B/{\alpha}]$ & receive type $B$ on $c$ \\
  \addlinespace
  $c : \tforall{\alpha}A $ & $c : A$ & $\erecvch{c}{[\alpha]} \semi P_{\alpha}$
  & $P_{\alpha}[B/\alpha]$ & receive type $B$ on $c$ \\
  & & $\esendch{c}{[B]} \semi Q$ & $Q$ & send type $B$ on $c$ \\
  \addlinespace
  $c : \one$ & --- & $\eclose{c}$
  & --- & send $\mi{close}$ on $c$ \\
  & & $\ewait{c} \semi Q$ & $Q$ & receive $\mi{close}$ on $c$ \\
  \bottomrule
  \end{tabular}
  \caption{Session types with operational description}
  \label{tab:language}
  \vspace{-2em}
  \end{table*}

%% file: context-free.tex
\section{Relationship to Context-Free Session Types}
\label{sec:context-free}

As ordinarily formulated, session types express communication protocols that can be described by regular languages~\cite{Thiemann16icfp}.
In particular, the type structure is necessarily tail recursive.
Context-free session types~(CFSTs) were introduced by Thiemann and Vascoconcelos~\cite{Thiemann16icfp} as a way to express a class of communication protocols that are not limited to tail recursion.
CFSTs express protocols that can be described by single-state, real-time DPDAs that use the empty stack acceptance criterion~\cite{Almeida20tacas,korenjak1966simple}.

Despite their name, the essence of CFSTs is not their connection to a particular subset of the (deterministic) context-free languages.
Rather, the essence of CFSTs is that session types are enriched to admit a notion of sequential composition.
\Polymorphicsessiontypes\ are strictly more expressive than CFSTs, in the sense that there exists a proper fragment of \polymorphicsessiontypes\ that is closed under a notion of sequential composition.
(In keeping with process algebras like ACP~\cite{Bergstra89}, we define a sequential composition to be an operation that satisfies the laws of a right-distributive monoid.)

Consider (up to $\alpha$,$\beta$,$\eta$-equivalence) the linear, tail functions from types to types with unary type constructors only:
\begin{equation*}
  S,T \Coloneqq
    \begin{array}[t]{@{}l@{}}
      \llam{\alpha}{\alpha} \mid \llam{\alpha}{V[\lapp{S}{\alpha}]}
      \mid \llam{\alpha}{\ichoice{\ell : \lapp{S_{\ell}}{\alpha}}_{\ell \in L}}
      \mid \llam{\alpha}{\echoice{\ell : \lapp{S_{\ell}}{\alpha}}_{\ell \in L}} \\
      \mathllap{\mid {}} \llam{\alpha}{A \tensor (\lapp{S}{\alpha})}
      \mid \llam{\alpha}{A \lolli (\lapp{S}{\alpha})}
    \end{array}
\end{equation*}
The linear, tail nature of these functions allows the type $\alpha$ to be thought of as a continuation type for the session.
The functions $S$ are closed under function composition, and the identity function, $\llam{\alpha}{\alpha}$, is included in this class of functions.
Moreover, because these functions are tail functions, composition right-distributes over the various logical connectives in the following sense:
\begin{align*}
  (\llam{\alpha}{V[\lapp{S}{\alpha}]}) \circ T
    &= \llam{\alpha}{V[\lapp{(S \circ T)}{\alpha}]}
  \\
  (\llam{\alpha}{\ichoice{\ell : \lapp{S_{\ell}}{\alpha}}_{\ell \in L}}) \circ T
    &= \llam{\alpha}{\ichoice{\ell : \lapp{(S_{\ell} \circ T)}{\alpha}}_{\ell \in L}}
  \\
  (\llam{\alpha}{A \tensor (\lapp{S}{\alpha})}) \circ T
    &= \llam{\alpha}{A \tensor (\lapp{(S \circ T)}{\alpha})}
\end{align*}
and similarly for $\echoiceop$ and $\lolli$.
Together with the monoid laws of function composition, these distributive properties justify defining sequential composition as $S;T = S \circ T$.

This suggests that although many details distinguish our work from CFSTs, nested session types cover the essence of sequential composition underlying context-free session types.
%
However, even 
stating a theorem that every CFST process can be translated into a well-typed process in our system of \polymorphicsessiontypes\
is difficult because the two type systems differ in many details:
we include $\tensor$ and $\lolli$ as session types, but CFSTs do not;
CFSTs use a complex kinding system to incorporate unrestricted session types and combine session types with ordinary function types;
the CFST system uses classical typing for session types and a procedure of type normalization, whereas our types are intuitionistic and do not rely on normalization;
and
the CFST typing rules are based on natural deduction, rather than the sequent calculus.
With all of these differences, a formal translation, theorem, and proof would not be very illuminating beyond the essence already described here.
Empirically, we can also give analogues of the published examples for CFSTs (see, e.g., the first two examples of Section~\ref{sec:examples}).

Finally, \polymorphicsessiontypes\ are strictly \emph{more} expressive than CFSTs.
Recall from Section~\ref{sec:overview} the language $L_3 = \{\mb{L} ^n \mb{a} \,\mb{R}^n \mb{a} \cup \mb{L}^n \mb{b} \,\mb{R}^n \mb{b} \mid n > 0\}$, which can be expressed using \polymorphicsessiontypes\ with \emph{two} type parameters used in an essential way.
Moreover, Korenjak and Hopcroft~\cite{korenjak1966simple} observe that this language cannot be recognized by a single-state, real-time DPDA that uses empty stack acceptance, and thus, CFSTs cannot express the language $L_3$.
More broadly, nested types allow for finitely many states
and acceptance by empty stack or final state, while CFSTs only allow a single state and empty stack acceptance.

%% file: implementation.tex
\section{Implementation}\label{sec:impl}

We have implemented a prototype for nested session
types and integrated it with the open-source Rast system~\cite{Das20FSCD}.
Rast (Resource-aware session types) is a programming language
which implements the intuitionistic version of session types~\cite{Caires10concur}
with support for arithmetic refinements~\cite{Das20CONCUR},
ergometric~\cite{Das18RAST} and temporal~\cite{Das18Temporal} types
for complexity analysis. Our prototype extension is implemented in Standard
ML (8011 lines of code) containing a lexer and parser (1214 lines),
a type checker (3001 lines) and an interpreter (201 lines) and
is well-documented.
The prototype is available in the Rast repository~\cite{RastBitBucket}.

\paragraph{\textbf{Syntax}}
A program contains a series of mutually recursive type and process
declarations and definitions, concretely written as
\begin{lstlisting}
  type V[x1]...[xk] = A
  decl f[x1]...[xk] : (c1 : A1) ... (cn : An) |- (c : A)
  proc c <- f[x] c1 ... cn = P
\end{lstlisting}
Type $V\indv{x}$ is represented in concrete syntax as \verb`V[x1]...[xk]`.
The first line is a \emph{type definition}, where $V$ is the type name
parameterized by type variables $x_1, \ldots, x_k$ and $A$ is its definition. The second
line is a \emph{process declaration}, where $f$ is the process name
(parameterized by type variables $x_1, \ldots, x_k$),
$(c_1 : A_1) \ldots (c_n : A_n)$ are the used channels and corresponding
types, while the offered channel is $c$ of type $A$. Finally, the last line is a
\emph{process definition} for the same process $f$ defined using the process
expression $P$. We use a hand-written lexer and shift-reduce parser to read an
input file and generate the corresponding abstract syntax tree of the program.
The reason to use a hand-written parser instead of a parser generator is to
anticipate the most common syntax errors that programmers make and respond with the
best possible error messages.

Once the program is parsed and its abstract syntax tree is extracted,
we perform a \emph{validity check} on it. This includes checking that
type definitions, and process declarations and definitions are closed
w.r.t. the type variables in scope.
To simplify and improve the efficiency of the type equality
algorithm, we also assign internal names to type subexpressions
parameterized over their free index variables. These internal names
are not visible to the programmer.

\paragraph{\textbf{Type Checking and Error Messages}}
The implementation is carefully designed to produce precise error messages.
To that end, we store the extent (source location) information with
the abstract syntax tree, and use it to highlight the source of the
error. We also follow a bi-directional type checking~\cite{Pierce00TOPLAS}
algorithm reconstructing intermediate types starting with the initial
types provided in the declaration. This helps us precisely identify
the source of the error. Another particularly helpful technique has been
\emph{type compression}. Whenever the type checker expands a type $V
\indv{A}$ defined as $V \indv{\alpha} = B$ to
$B[\overline{A}/\overline{\alpha}]$ we record a reverse mapping from
$B[\overline{A}/\overline{\alpha}]$ to $V \indv{\alpha}$.  When printing types
for error messages this mapping is consulted, and complex types may be
compressed to much simpler forms, greatly aiding readability of
error messages.

%% file: examples.tex
\section{More Examples}
\label{sec:examples}


\paragraph{\textbf{Expression Server}}
We adapt the example of an arithmetic expression 
from prior work on context-free session types~\cite{Thiemann16icfp}.
The type of the server is defined as
\begin{verbatim}
  type bin = +{ b0 : bin, b1 : bin, $ : 1 }
  type tm[K] = +{ const : bin * K,
                  add : tm[tm[K]],
                  double : tm[K] }
\end{verbatim}
The type \verb`bin` represents a constant binary natural number. A
process \emph{providing} a binary number sends a stream of bits, \verb`b0` and \verb`b1`,
starting with the least significant bit and eventually terminated by \verb`$`.

An arithmetic term, parameterized by continuation type \verb`K` can
have one of three forms: a constant, the sum of two terms,
or the double of a term. Consequently, the type \verb`tm[K]` ensures that
 a process providing \verb`tm[K]` is a \emph{well-formed term}: it
either sends the \verb`const` label followed by sending a constant binary
number of type \verb`bin` and continues with type \verb`K`; or it sends
the \verb`add` label and continues with \verb`tm[tm[K]]`, where
the two terms denote the two summands; or it sends the \verb`double`
label and continues with \verb`tm[K]`.
 In particular, the continuation type
\verb`tm[tm[K]]` in the \verb`add` branch enforces that the process must send exactly
two summands for sums.

As a first illustration, consider two binary constants $a$ and $b$,
and suppose that we want to create the expression $a+2b$. We can issue
commands to the expression server in a \emph{prefix notation}
to obtain $a+2b$, as shown in the following \verb`exp[K]` process, which is parameterized by a continuation type \verb`K`.
\begin{verbatim}
  decl exp[K] : (a : bin) (b : bin) (k : K) |- (e : tm[K])
  proc e <- exp[K] a b k =
    e.add ; e.const ; send e a ; % (b:bin) (k:K) |- (e : tm[K])
    e.double ; e.const ; send e b ; % (k:K) |- (e : K)
    e <-> k
\end{verbatim}
In prefix notation, $a+2b$ would be written $+ \; (a) \; (2 \; b)$,
which is exactly the form followed by the \verb`exp` process:
The process sends \verb`add`, followed by \verb`const` and the number
\verb`a`, followed by \verb`double`, \verb`const`, and \verb`b`.
Finally, the process continues at type \verb`K` by forwarding \verb`k` to \verb`e`
(intermediate typing contexts on the right).


To evaluate a term, we can define an \verb`eval` process, parameterized by type \verb`K`:

\begin{verbatim}
decl eval[K] : (t : tm[K]) |- (v : bin * K)
\end{verbatim}
The \verb`eval` process uses channel
\verb`t : tm[K]` as argument,
and offers \verb`v : bin * K`. The process evaluates
term \verb`t` and sends its binary value along \verb`v`.
\begin{verbatim}
  decl eval[K] : (t : tm[K]) |- (v : bin * K)
  proc v <- eval[K] t =
    case t (
      const =>                % (t : bin * K) |- (v : bin * K)
        n <- recv t ;         % (n : bin) (t : K) |- (v : bin * K)
        send v n ; v <-> t
    | add =>                  % (t : tm[tm[K]]) |- (v : bin * K)
        v1 <- eval[tm[K]] t ; % (v1 : bin * tm[K]) |- (v : bin * K)
        n1 <- recv v1 ;       % (n1 : bin) (v1 : tm[K]) |- (v : bin * K)
        v2 <- eval[K] v1 ;    % (n1 : bin) (v2 : bin * K) |- (v : bin * K)
        n2 <- recv v2 ;       % (n1 : bin) (n2 : bin) (v2 : K) |- (v : bin * K)
        n <- plus n1 n2 ;     % (n : bin) (v2 : K) |- (v : bin * K)
        send v n ; v <-> v2
    | double =>               % (t : tm[K]) |- (v : bin * K)
        v1 <- eval[K] t ;     % (v1 : bin * K) |- (v : bin * K)
        n1 <- recv v1 ;       % (n1 : bin) (v1 : K) |- (v : bin * K)
        n <- double n1 ;      % (n : bin) (v1 : K) |- (v : bin * K)
        send v n ; v <-> v1
    )
  \end{verbatim}
  Intuitively, the process evaluates
  term \verb`t` and sends its binary value along \verb`v`. If \verb`t` is a constant, then
  \verb`eval[K]` receives the constant $n$, sends it along $v$ and forwards.
  
  The interesting case is the \verb`add` branch. We evaluate the first
  summand by spawning a new \verb`eval[K]` process on \verb`t`. Note that since
  the type of \verb`t` (indicated on the right) is \verb`tm[tm[K]]`
  and hence, the recursive call to \verb`eval` is at parameter  \verb`tm[K]`.
  This is in contrast with \emph{nominal polymorphism} in functional
  programming languages, where the recursive call must also be at parameter
  \verb`K`. We store the value of the first summand at channel \verb`n1 : bin`.
  Then, we continue to evaluate the second summand by calling \verb`eval[K]`
  on \verb`t` again and storing its value in \verb`n2 : bin`. Finally, we add
  \verb`n1` and \verb`n2` by calling the \verb`plus` process, and send the result
  \verb`bin` along \verb`v`. We follow a similar approach for the \verb`double`
  branch.

\paragraph{\textbf{Serializing binary trees}}

Another example from \cite{Thiemann16icfp} is serializing binary trees.
Here we adapt that example to our system.
Binary trees can be described by:
\begin{verbatim}
type Tree[a] = +{ node : Tree[a] * a * Tree[a] , leaf : 1 }
\end{verbatim}
These trees are polymorphic in the type of data stored at each internal node.
A tree is either an internal node or a leaf, with the internal nodes storing channels that emit the left subtree, data, and right subtree.
To help in creating trees, we can define the following processes.
\begin{verbatim}
decl leaf[a] : . |- (t : Tree[a])
proc t <- leaf[a] =
  t.leaf ; close t

decl node[a] : (l : Tree[a]) (x : a) (r : Tree[a]) |- (t : Tree[a])
proc t <- node[a] l x r =
  t.node ; send t l ; send t x ; t <-> r
\end{verbatim}

Owing to the multiple channels stored at each node, these trees do not exist in a serial form.
We can, however, use a different type to represent serialized trees:
\begin{verbatim}
type STree[a][K] = +{ nd : STree[a][a * STree[a][K]] , lf : K }
\end{verbatim}
A serialized tree is a stream of node and leaf labels, \verb`nd` and \verb`lf`, parameterized by a continuation type \verb`K`.
Like \verb`add` in the expression server, the label \verb`nd` continues with type \verb`STree[a][a * STree[K]]`: the label \verb`nd` is followed by the serialized left subtree, which itself continues by sending the data stored at the internal node and then the serialized right subtree, which continues with type \verb`K`.\footnote{The presence of \texttt{a *} means that this is not a true serialization because it sends a separate channel along which the data of type \texttt{a} is emitted.
But there is no uniform mechanism for serializing polymorphic data, so this is as close to a true serialization as possible.
Concrete instances of type \texttt{Tree} with, say, data of base type \texttt{int} could be given a true serialization by ``inlining'' the data of type \texttt{int} in the serialization.}

Using these types, it is relatively straightforward to implement processes that serialize and deserialize such trees.
The process \verb`serialize` can be declared with:
\begin{verbatim}
decl serialize[a][K] : (t : Tree[a]) (k : K) |- (s : STree[a][K])
\end{verbatim}
This process uses channels \verb`t` and \verb`k` that hold the tree and continuation, and offers that tree's serialization along channel \verb`s`.
The complete code for \verb`serialize` (and a helper process) is:
\begin{verbatim}
decl output[a][b] : (y : a) (x' : b) |- (x : a * b)
proc x <- output[a][b] y x' =
  send x y ; x <-> x'

decl serialize[a][K] : (t : Tree[a]) (k : K) |- (s : STree[a][K])
proc s <- serialize[a][K] t k =
  case t (
    leaf =>             % (t:1) (k:K) |- (s:STree[a][K])
      s.lf ;            % (t:1) (k:K) |- (s:K)
      wait t ; s <-> k
  | node =>
        % (t : Tree[a]*a*Tree[a]) (k:K) |- (s:STree[a][K])
      l <- recv t ;
        % (l:Tree[a]) (t : a*Tree[a]) (k:K) |- (s:STree[a][K])
      x <- recv t ;
        % (l:Tree[a]) (x:a) (t:Tree[a]) (k:K) |- (s:STree[a][K])
      sr <- serialize[a][K] t k ;
        % (l:Tree[a]) (x:a) (sr:STree[a][K]) |- (s:STree[a][K])
      sx <- output[a][STree[a][K]] x sr ;
        % (l:Tree[a]) (sx : a*STree[a][K]) |- (s:STree[a][K])
      s.nd ;
        % (l:Tree[a]) (sx : a*STree[a][K]) |- (s:STree[a][K])
      s <- serialize[a][a * STree[a][K]] l sx )
\end{verbatim}
If the tree is only a leaf, then the process forwards to the continuation.
Otherwise, if the tree begins with a node, then the serialization begins with \verb`nd`.
A recursive call to \verb`serialize` serves to serialize the right subtree with the given continuation.
A subsequent recursive call serializes the left subtree with the data together with the right subtree's serialization as the new continuation.

It is also possible to implement a process for deserializing trees:
\begin{verbatim}
decl deserialize[a][K] : (s : STree[a][K]) |- (tk : Tree[a] * K)
proc tk <- deserialize[a][K] s =
  case s (
    lf =>
        % (s : K) |- (tk : Tree[a] * K)
      t <- leaf[a] ;
        % (t : Tree[a]) (s : K) |- (tk : Tree[a] * K)
      send tk t ;
        % (s : K) |- (tk : K)
      tk <-> s
  | nd =>
        % (s : STree[a][a * STree[a][K]]) |- (tk : Tree[a] * K)
      lk <- deserialize[a][a * STree[a][K]] s ;
        % (lk : Tree[a] * (a * STree[a][K])) |- (tk : Tree[a] * K)
      l <- recv lk ;
        % (l:Tree[a]) (lk : a * STree[a][K]) |- (tk : Tree[a] * K)
      x <- recv lk ;
        % (l:Tree[a]) (x:a) (lk:STree[a][K]) |- (tk : Tree[a] * K)
      rk <- deserialize[a][K] lk ;
        % (l:Tree[a]) (x:a) (rk : Tree[a] * K) |- (tk : Tree[a] * K)
      r <- recv rk ;
        % (l:Tree[a]) (x:a) (r:Tree[a]) (rk:K) |- (tk : Tree[a] * K)
      t <- node[a] l x r ;
        % (t:Tree[a]) (rk:K) |- (tk : Tree[a] * K)
      send tk t ;
        % (rk:K) |- (tk : K)
      tk <-> rk )
\end{verbatim}

\paragraph{\textbf{Generalized tries for binary trees}}

Using nested types in Haskell, prior work~\cite{Hinze10jfp} describes an implementation of generalized tries that represent mappings on binary trees.
Our type system is expressive enough to represent such generalized tries.
We can reuse the type \verb`Tree[a]` of binary trees given above.
The type \verb`Trie[a][b]` describes tries that represent mappings from \verb`Tree[a]` to type \verb`b`:
\begin{verbatim}
type Trie[a][b] = &{ lookup_leaf : b ,
                     lookup_node : Trie[a][a -o Trie[a][b]] }
\end{verbatim}

A process for looking up a tree in such tries can be declared by:
\begin{verbatim}
decl lookup_tree[a][b] : (m : Trie[a][b]) (t : Tree[a]) |- (v : b)
\end{verbatim}
To lookup a tree in a trie, first determine whether that tree is a \verb`leaf` or a \verb`node`.
If the tree is a \verb`leaf`, then sending \verb`lookup_leaf` to the trie will return the value of type \verb`b` associated with that tree in the trie.

Otherwise, if the tree is a \verb`node`, then sending \verb`lookup_node` to the trie results in a trie of type \verb`Trie[a][a -o Trie[a][b]]` that represents a mapping from left subtrees to type \verb`a -o Trie[a][b]`.
We then lookup the left subtree in this trie, resulting in a process of type \verb`a -o Trie[a][b]` to which we send the data stored at our original tree's root.  That results in a trie of type \verb`Trie[a][b]` that represents a mapping from right subtrees to type \verb`b`.
Therefore, we finally lookup the right subtree in this new trie and obtain a result of type \verb`b`, as desired.

We can define a process that constructs a trie from a function on trees:
\begin{verbatim}
decl build_trie[a][b] : (f : Tree[a] -o b) |- (m : Trie[a][b])
\end{verbatim}
Both \verb`lookup_tree` and \verb`build_trie` can be seen as analogues to \verb`deserialize` and \verb`serialize`, respectively, converting a lower-level representation to a higher-level representation and vice versa.
These types and declarations mean that tries represent total mappings; partial mappings are also possible, at the expense of some additional complexity.

All our examples have been implemented and type checked in the open-source Rast
repository~\cite{RastBitBucket}. We have also further implemented the standard
polymorphic data structures such as lists, stacks and queues.

%% file: related-work.tex
\section{Further Related Work}

To our knowledge, our work is the first proposal of polymorphic
recursion using nested type definitions in session
types. Thiemann and Vasconcelos~\cite{Thiemann16icfp} use polymorphic recursion to update the
channel between successive recursive calls but do not allow type
constructors or nested types. An algorithm to check type
equivalence for the non-polymorphic fragment of context-free session
types has been proposed by Almeida et al.~\cite{Almeida20tacas}.

Other forms of polymorphic session types have also been considered in the
literature.  Gay~\cite{Gay08} studies bounded polymorphism associated
with branch and choice types in the presence of subtyping.  He
mentions recursive types (which are used in some examples) as future
work, but does not mention parametric type definitions or nested
types. Bono and Padovani~\cite{Bono11,Bono12} propose (bounded) polymorphism to type
the endpoints in copyless message-passing programs inspired by session
types, but they do not have nested types. Following Kobayashi's
approach~\cite{Kobayashi03Book}, Dardha et al.~\cite{DardhaGS17} provide
an encoding of session types relying on linear and variant types and
present an extension to enable parametric and bounded polymorphism (to
which recursive types were added separately~\cite{Dardha14beat}) but
not parametric type definitions nor nested
types. Caires et al.~\cite{PolyESOP13} and Perez et al.~\cite{Perez14ic} 
provide behavioral polymorphism
and a relational parametricity principle for session types, but
without recursive types or type constructors.

Nested session types bear important similarities with
first-order cyclic terms, as observed by
Jan\v{c}ar. Jan\v{c}ar~\cite{Jancar10} proves that the trace equivalence problem
of first-order grammars is decidable, following the original ideas by
Stirling for the language equality problem in deterministic pushdown
automata~\cite{Stirling01tcs}. These ideas were also reformulated
by S\'enizergues~\cite{Senizergues02}. 
Henry and S\'enizergues~\cite{henry2013lalblc} proposed the only
practical algorithm to decide the language equivalence problem on
deterministic pushdown automata that we are aware of. Preliminary 
experiments show that
such a generic implementation, even if complete in theory, is a poor
match for the demands made by our type checker.

%% file: conclusion.tex
\section{Conclusion}

\Polymorphicsessiontypes\ extend binary session types with parameterized type 
definitions. This extension enables us to express 
polymorphic data structures just as naturally as in functional languages.  
The proposed types are able to capture sequences of communication actions
described by deterministic context-free languages recognized by deterministic
pushdown  automata with several states, that accept by empty stack or by final
state. In this setting, we show that type equality is decidable.
To offset the complexity of type equality, we give
a practical type equality algorithm that is sound, efficient, but incomplete.

In the future, we are planning to explore subtyping for nested types.
In particular, since the language inclusion problem for simple
languages~\cite{Friedman76} is undecidable, we believe subtyping can
be reduced to inclusion and would also be undecidable. Despite this
negative result, it would be interesting to design an
algorithm to approximate subtyping. That would significantly increase the
programs that can be type checked in the system. In another direction, 
since Rast~\cite{Das20FSCD}
supports arithmetic refinements for lightweight verification, it would be
interesting to explore how refinements interact with polymorphic type
parameters, namely in the presence of subtyping. We would also like
to explore examples where the current type equality is not adequate.
Finally, protocols in distributed algorithms such as consensus or leader election (Raft, Paxos, etc.)\ depend on unbounded memory and cannot usually be expressed with finite control structure.
In future work, we would like to see if these protocols can be expressed with nested session types.